\documentclass[11pt]{article}
  \usepackage{amsthm}
  \usepackage{amssymb}
  \usepackage{amsmath}
  \usepackage{eucal}
  \usepackage{hyperref}
  \usepackage{mathrsfs}
  \usepackage{yfonts}
  \usepackage{bbm}
  \usepackage[T1]{fontenc}
  \usepackage{tikz}
\usetikzlibrary{decorations.pathmorphing,decorations.markings}
  \usepackage{bbm}
  \usepackage{geometry}
  \usepackage[all]{xy}

\newcommand{\M}{\mathcal{M}} 
\newcommand{\MM}{\mathbbmss{M}} 

\newcommand{\al}{\alpha}
\newcommand{\bet}{\beta}

\newcommand{\de}{\delta}
\newcommand{\De}{\Delta}

\newcommand{\La}{\Lambda}

\newcommand{\ph}{\varphi}

\newcommand{\be}{\begin{equation}}
\newcommand{\ee}{\end{equation}}

\newcommand{\id}{\mathrm{id}} 

\newcommand{\NN}{\mathbbmss{N}} 
\newcommand{\ZZ}{\mathbbmss{Z}}
\newcommand{\RR}{\mathbbmss{R}} 
\newcommand{\CC}{\mathbbmss{C}} 
\newcommand{\supp}{\mathrm{supp}}
\newcommand{\WF}{\mathrm{WF}} 

\newcommand{\Dcal}{\mathcal{D}}
\newcommand{\Ecal}{\mathcal{E}} 
\newcommand{\Fcal}{\mathcal{F}} 
\newcommand{\Floc}{\mathcal{F}_{\text{loc}}} 
\newcommand{\Gcal}{\mathcal{G}} 

\newcommand{\Hcal}{\mathcal{H}}

\newcommand{\Ocal}{\mathcal{O}}

\newcommand{\Rcal}{\mathscr{R}}
\newcommand{\Scal}{\mathcal{S}}
\newcommand{\Tcal}{\mathcal{T}}

\newcommand{\Zcal}{\mathcal{Z}}

\newcommand{\sd}{\mathrm{sd}}

\newcommand{\ml}{{\mathrm{ml}}}
\newcommand{\reg}{\mathrm{reg}}

\DeclareMathAlphabet{\mathpzc}{OT1}{pzc}{m}{it}
\newcommand{\wcal}{\mathpzc{w}}

\newcommand{\bzeta}{{\boldsymbol{\zeta}}}

\newcommand{\pp}{\mathrm{pp}}		
\newcommand{\rp}{\mathrm{rp}}		
\renewcommand{\Re}{\mathrm{Re}}		
\newcommand{\MS}{\mathrm{MS}}		
\newcommand{\loc}{\mathrm{loc}}
\newcommand{\Lin}{\mathrm{Lin}}		
\newcommand{\oc}{{\,\bigcirc\hspace{-0.65em}\mathsf{c}\;}}		
\newcommand{\Halg}{\mathsf{H}}		
\newcommand{\Ap}{\mathcal{A}}		
\newcommand{\Oalg}{\mathcal{O}}

\newcommand{\Falg}{\mathsf{F}}	
\newcommand{\T}{\cdot_{{}^\Tcal}}

\newcommand{\TT}{\Tcal}

\tikzset{
photon/.style={decorate, decoration={snake,amplitude=4pt, segment length=5pt}, draw=red},
particle/.style={draw=blue, postaction={decorate}, decoration={markings,mark=at position .5 with {\arrow[draw=blue]{>}}}},
antiparticle/.style={draw=blue, postaction={decorate}, decoration={markings,mark=at position .5 with {\arrow[draw=blue]{<}}}},
gluon/.style={decorate, draw=green, decoration={coil,amplitude=4pt, segment length=5pt}}
}
\newcommand{\Catseye}{\begin{tikzpicture}[thick,scale=1.5]
\useasboundingbox (0,-0.1) rectangle (1,0.4);
\filldraw (0,0) circle (1pt);
\filldraw (0.4,0.13) circle (1pt);
\filldraw (0.4,-0.13) circle (1pt);
\filldraw (0.8,0) circle (1pt);
\draw (0,0) .. controls (0.18,0.4) and (0.62,0.4) .. (0.8,0);
\draw (0,0) edge [out=35,in=145] node[above] {} (0.8,0);
\draw (0,0) edge [out=-35,in=-145] node[above] {} (0.8,0);
\draw (0.4,0.13) edge [out=-25,in=30] node[above] {} (0.4,-0.13);
\draw (0.4,0.13) edge [out=-145,in=140] node[above] {} (0.4,-0.13);
\end{tikzpicture} }
\newcommand{\FGHtwoF}{\begin{tikzpicture}[thick,scale=1.5]
\useasboundingbox (-0.1,0.1) rectangle (0.5,0.4);
\filldraw (0,0) circle (1.3pt);
\filldraw (0.4,0) circle (1.3pt);
\filldraw (0.2,0.3) circle (1.3pt);
\draw (0.4,0) edge [out=75,in=-5] node[above] {} (0.2,0.3);
\draw (0.4,0) edge [out=-185,in=-105] node[above] {} (0.2,0.3);
\end{tikzpicture} }
\newcommand{\FGtwoHF}{\begin{tikzpicture}[thick,scale=1.5]
\useasboundingbox (-0.1,0.1) rectangle (0.5,0.4);
\filldraw (0,0) circle (1.3pt);
\filldraw (0.4,0) circle (1.3pt);
\filldraw (0.2,0.3) circle (1.3pt);
\draw (0,0) edge [out=45,in=135] node[above] {} (0.4,0);
\draw (0,0) edge [out=-45,in=-135] node[above] {} (0.4,0);
\end{tikzpicture} }
\newcommand{\FtwoGHF}{\begin{tikzpicture}[thick,scale=1.5]
\useasboundingbox (-0.1,0.1) rectangle (0.5,0.4);
\filldraw (0,0) circle (1.3pt);
\filldraw (0.4,0) circle (1.3pt);
\filldraw (0.2,0.3) circle (1.3pt);
\draw (0,0) edge [out=105,in=185] node[above] {} (0.2,0.3);
\draw (0,0) edge [out=5,in=-75] node[above] {} (0.2,0.3);
\end{tikzpicture} }
\newcommand{\FGoneHoneF}{\begin{tikzpicture}[thick,scale=1.5]
\useasboundingbox (-0.1,0.1) rectangle (0.5,0.4);
\filldraw (0,0) circle (1.3pt);
\filldraw (0.4,0) circle (1.3pt);
\filldraw (0.2,0.3) circle (1.3pt);
\draw (0.4,0) -- (0.2,0.3);
\draw (0.4,0) -- (0,0);
\end{tikzpicture} }
\newcommand{\FoneGoneHF}{\begin{tikzpicture}[thick,scale=1.5]
\useasboundingbox (-0.1,0.1) rectangle (0.5,0.4);
\filldraw (0,0) circle (1.3pt);
\filldraw (0.4,0) circle (1.3pt);
\filldraw (0.2,0.3) circle (1.3pt);
\draw (0,0) -- (0.2,0.3);
\draw (0.4,0) -- (0,0);
\end{tikzpicture} }
\newcommand{\FoneGHoneF}{\begin{tikzpicture}[thick,scale=1.5]
\useasboundingbox (-0.1,0.1) rectangle (0.5,0.4);
\filldraw (0,0) circle (1.3pt);
\filldraw (0.4,0) circle (1.3pt);
\filldraw (0.2,0.3) circle (1.3pt);
\draw (0,0) -- (0.2,0.3);
\draw (0.4,0) -- (0.2,0.3);
\end{tikzpicture} }
\newcommand{\FGHoneF}{\begin{tikzpicture}[thick,scale=1.5]
\useasboundingbox (-0.1,0.1) rectangle (0.5,0.4);
\filldraw (0,0) circle (1.3pt);
\filldraw (0.4,0) circle (1.3pt);
\filldraw (0.2,0.3) circle (1.3pt);
\draw (0,0) -- (0.2,0.3);
\end{tikzpicture} }
\newcommand{\FGoneHF}{\begin{tikzpicture}[thick,scale=1.5]
\useasboundingbox (-0.1,0.1) rectangle (0.5,0.4);
\filldraw (0,0) circle (1.3pt);
\filldraw (0.4,0) circle (1.3pt);
\filldraw (0.2,0.3) circle (1.3pt);
\draw (0,0) -- (0.4,0);
\end{tikzpicture} }
\newcommand{\FoneGHF}{\begin{tikzpicture}[thick,scale=1.5]
\useasboundingbox (-0.1,0.1) rectangle (0.5,0.4);
\filldraw (0,0) circle (1.3pt);
\filldraw (0.4,0) circle (1.3pt);
\filldraw (0.2,0.3) circle (1.3pt);
\draw (0.2,0.3) -- (0.4,0);
\end{tikzpicture} }
\newcommand{\FGH}{\begin{tikzpicture}[thick,scale=1.5]
\useasboundingbox (-0.1,0.1) rectangle (0.5,0.4);
\filldraw (0,0) circle (1.3pt);
\filldraw (0.4,0) circle (1.3pt);
\filldraw (0.2,0.3) circle (1.3pt);
\end{tikzpicture} }

 \author{\null\\Michael D\"utsch$^{(2)}$, Klaus Fredenhagen$^{(1)}$, Kai Johannes Keller$^{(1)}$, Katarzyna Rejzner$^{(3)}$ \\
  \null\\
  \null\\
        \small{$^{(1)}$ 2. Inst. f. Theoretische Physik, Universit\"at Hamburg,}\\
    \small{Luruper Chaussee 149, D-22761 Hamburg, Germany}\\
\small{$^{(2)}$ Max Planck Institute for Mathematics in the Sciences,}\\
\small{Inselstrasse 22, D-04103 Leipzig}\\
\small{$^{(3)}$Mathematics Department, University of York}\\
\small{York YO10 5DD}\\
\small{\texttt{michael.duetsch@theorie.physik.uni-goettingen.de,
klaus.fredenhagen@desy.de,}}\\ 
\small{\texttt{kai.johannes.keller@desy.de, kasia.rejzner@york.ac.uk }}}

  \title{Dimensional Regularization in Position Space, 
and a Forest Formula for Epstein-Glaser Renormalization}

\begin{document}
\maketitle

\begin{abstract}
We reformulate dimensional regularization as a regularization method in position space and show 
that it can be used to give a closed expression for the renormalized time-ordered products as 
solutions to the induction scheme of Epstein-Glaser. This closed expression, which we call the 
Epstein-Glaser Forest Formula, is analogous to Zimmermann's Forest Formula for BPH renormalization.
For scalar fields the resulting renormalization 
method is always applicable, we compute several examples. We also analyze the Hopf algebraic aspects 
of the combinatorics. Our starting point is the 
Main Theorem of Renormalization
of Stora and Popineau and the arising renormalization group as originally defined by St{\"u}ckelberg and Petermann. 
\end{abstract}

  \theoremstyle{plain}
  \newtheorem{df}{Definition}[section]
  \newtheorem{thm}[df]{Theorem}
  \newtheorem{prop}[df]{Proposition}
  \newtheorem{cor}[df]{Corollary}
  \newtheorem{lemma}[df]{Lemma}
  
  \theoremstyle{plain}
  \newtheorem*{Main}{Main Theorem}
  \newtheorem*{MainT}{Main Technical Theorem}

  \theoremstyle{definition}
  \newtheorem{rem}[df]{Remark}
  \newtheorem{example}[df]{Example}

 \theoremstyle{definition}
  \newtheorem{ass}{\underline{\textit{Assumption}}}[section]

{\small
\tableofcontents
\markboth{Contents}{Contents}}
\section{Introduction}
Renormalized perturbative quantum field theory describes large parts of physics, in particular particle physics, with good, and sometimes spectacular precision. It is, however, a conceptually and technically complicated subject, and it required hard and ingenious work to put the original treatment of Tomonaga, Schwinger, Feynman and Dyson on solid grounds. This was achieved, by Bogoliubov, Parasiuk, Hepp, Zimmermann, Epstein, Glaser, Steinmann and others, in a twenty years struggle, and the finally reached state of the art is nicely documented in the proceedings of the Erice school 1975 dedicated to renormalization \cite{VW76}. Main highlights are the Forest Formula of Zimmermann \cite{Zim69} which solves the recursion relations of the Bogoliubov-Parasiuk-Hepp (BPH) method \cite{BP57,BS59,Hep66}, the causal method of Epstein-Glaser (EG)\cite{EG73}, elaborating on older attempts of St\"uckelberg \cite{SR50} and Bogoliubov \cite{BP57,BS59}, and the method of retarded products by Steinmann \cite{Ste71}.

In spite of the fact that highly nontrivial mathematical methods were used (and to some extent, invented), the theory of perturbative renormalization had, for several decades, less impact on mathematics than it deserved%
\footnote{See however the work induced by Polchinski's version \cite{Pol84} of the Wilsonian renormalization group: \cite{KKS91,KK91,KK92,KK99}.}. 
This changed recently, induced by the observation of Kreimer \cite{Kre98}
that the BPH recursion relations may be understood in terms of Hopf algebras. It culminated in the Connes-Kreimer theory of renormalization \cite{CK00,CK01} and initiated a broad interest of mathematicians in perturbative quantum field theory.

In the present formulation (see, e.g., the book of Connes and Marcolli \cite{CM07}) the theory is based on the method of dimensional regularization, and on the combinatorics of Zimmermann's Forest Formula.
Dimensional regularization was invented simultaneously by Bollini and Giambiagi \cite{BG72a} and by 't Hooft and Veltman \cite{tHV72}. It relies on the fact that after parametrizing Feynman integrals by Schwinger or by Feynman parameters the momentum space integrals can be performed, and it remains an integral over the parameters whose integrand depends on the spacetime dimension. Formally, one can replace the spacetime dimension by an arbitrary complex number $d$. The resulting integral exists on a certain domain of the complex plane; moreover, it can be extended to a meromorphic function on the whole complex plane. A finite value at the physical dimension is obtained by subtracting the principal part of the Laurent series. It is, however, not a priori clear that this procedure is physically meaningful. A similar situation is present in the so-called $\zeta$-function renormalization where the deeper reasons for the spectacular successes are not well understood (see, e.g. \cite{Elizalde:1994gf}). In the case of dimensional regularization the situation was clarified by the analysis of Breitenlohner and Maison \cite{BM77a,BM77b,BM77c} who showed how the combinatorics of Zimmermann's Forest Formula can be adapted to dimensional regularization.

Dimensional regularization turned out to be very effective for practical calculations, in particular due to the fact that gauge invariance is not broken during the renormalization process. Its conceptual basis is, however, not very transparent.

Quite the contrary is true for EG renormalization. This method is based on the observation that time-ordered products of local fields can, up to coinciding points, be performed as operator products. The latter are well defined in the sense of operator valued distributions on Fock space, as shown originally by G\aa{}rding and Wightman \cite{GW64}. Thus we know, from the very beginning, the time-ordered products everywhere up to coinciding points.  Then, using an induction process, we can prove that the time-ordered product of $n$ local fields is, outside the thin diagonal (where all arguments coincide), uniquely determined in terms of lower order time-ordered products. The induction step amounts to an extension of a distribution in the relative coordinates which is defined for test functions which vanish in the neighborhood of the origin, to all test functions. The latter process is ambiguous, and the ambiguities correspond to the freedom of adding finite counter terms to the interaction Lagrangean.

The nice features of EG renormalization, which in particular allow renormalization on generic globally hyperbolic spacetimes \cite{BF00,HW01,HW02} are, unfortunately, connected with the difficulty of carrying through explicit calculations. Nevertheless, quite a number of  computations have been performed within this framework (see, e.g. \cite{Sch89,GraciaBondiaLazzarini2003,DF04}). There remains, however, an impression that, essentially, one needs a new idea for every new calculation.
The main purpose of this paper is to develop a method for practical calculations which is always applicable. 

A similar problem arises when one tries to analyze the combinatorics of EG renormalization. There are interesting attempts in this direction, see e.g. \cite{GBL00,Pin00b,BK05,BBK09}, but the obtained picture is not yet completely satisfactory. 

The method we describe in this paper is based on the Main Theorem of Renormalization. This theorem, originally formulated in an unpublished preprint by Stora and Popineau \cite{SP82}, was later generalized and improved, in particular by Pinter \cite{Pin01}. Its final version, which relies heavily on a proof of Stora's ``Action Ward Identity'' \cite{DF04,DF07}, was obtained in \cite{HW03,DF04} and was then further analyzed  in \cite{BDF09}. The main statement is that the ambiguity of associating a perturbative quantum field theory to an interaction Lagrangean is described in terms of a group of formal diffeomorphisms (tangent to the identity) on the space of interaction Lagrangeans. Such a group also appears in the work of Connes and Kreimer, and one of the aims of the present paper is to understand the relations between the two frameworks.

The first insight is that, due to the Main Theorem of Renormalization, the combinatorics of finite renormalizations derives from an iterated application of the chain rule. In fact this combinatorics was investigated long ago by Fa\`a di Bruno \cite{FdB1855}, and the relation to the combinatorics observed in perturbation theory was nicely described in \cite{FGB05}. 

One may now, after performing the construction of the theory by the EG procedure, introduce a regularization and  ask for a renormalization group element which subtracts the counter terms in such a way that the regularized theory converges to the given theory. If the regularized theory depends meromorphically on the regularization parameter, it is clear that the principal part of the Laurent series of the regularized theory must coincide with that of the counter terms. It is now tempting to identify the counter terms with the principal part and to define a new theory corresponding to minimal subtraction.

The arising method works in an inductive way by proceeding order by order and inserting the results of the lower orders into the calculations for the next order. An obvious question is whether the result at 
$n$th order can be obtained directly, as in Zimmermann's Forest Formula. We will derive such a formula in the framework of Epstein-Glaser renormalization.

Dimensional regularization in position space amounts essentially to a change of the order of the Bessel functions defining the propagators. 
Such a procedure was first proposed by Bollini and Giambiagi \cite{BG96} and was also tested in several examples by \cite{GKP07}. 
It can be viewed as a particular 'analytic regularization', as introduced by Speer in the context of BPHZ-renormalization 
long ago \cite{Speer1971}, and applied to EG renormalization by Hollands \cite{Hollands2007}.
A different approach had been taken by Rosen and Wright \cite{RoW}: they implement dimensional regularization in $x$-space
 by making replacements on the level of the position space Feynmann rules. In particular, the spacetime coordinate $x$ is replaced by $X=(x,\hat{x})$, where $\hat{x}$ is a formal parameter corresponding to the ``integration over the complex dimension''. This approach, similar to the one taken by Breitenlohner and Maison \cite{BM77a}, seems to be very formal, since it is not clear if the algebraic relations postulated for the formal symbol can be fulfilled in any concrete model.

With the procedure of \cite{BG96}, which we adopt in this paper,
the regularized theory can be uniquely defined as a meromorphic function of the 
regularization parameter (which is called the ``dimension'' in the physics literature). Its analyticity 
property directly follows from the analytic dependence of Bessel functions on their 
order. The analytic continuation to a meromorphic function with a pole at the physical value of the regularization parameter can be performed by exploiting homogeneity properties. This 
appears very clearly in 
{\it massless} theories.
There, the inductive Epstein-Glaser construction of time-ordered products can 
be traced back to the extension of homogeneously or, for terms with divergent subdiagrams, almost 
homogeneously scaling distributions $t\in\Dcal^\prime(\mathbb{R}^l\setminus\{0\})$ to almost homogeneously scaling
distributions $\dot t\in\Dcal^\prime(\mathbb{R}^l)$. Existence and uniqueness of such extensions
are classified in Prop.~3.3 of \cite{DF04} (related theorems, which are precursors of this proposition, 
can be found in \cite[Thm.~3.2.3]{Hoer03} and \cite{HW02}):
\begin{prop}\label{alm-hom-scal} 
Let $t\in \Dcal^\prime(\mathbb{R}^l\setminus\{0\})$ scale
almost homogeneously with degree $\kappa\in\mathbb{C}$ and power $N\in\mathbb{N}_0$, i.e.
\begin{equation}\label{almosthomscal}
\sum_{r=1}^l(m_{x_r}\partial_{x_r}+\kappa)^{(N+1)}t=0
\end{equation}
where $N$ is the minimal natural number with this property and $m_{x_r}$ is a multiplication with the function $x\mapsto x_r$. (For $N=0$ the scaling is homogeneous.)
Then $t$ has an extension $\dot t\in \Dcal^\prime(\mathbb{R}^l)$ which also scales almost homogeneously with
degree $\kappa$ and with power $\dot N\geq N$.
\begin{itemize}
\item For $\kappa\not\in\mathbb{N}_0+l$ it holds: $\dot t$ is uniquely determined and it is $\dot N=N$;
\item for $\kappa\in\mathbb{N}_0+l$ we have: $\dot t$ is non-unique and $\dot N\in\{N,\,N+1\}$.
\end{itemize}
\end{prop}
In the {\it unregularized} theory we have $\kappa\in\mathbb{Z}$, and usually there are terms with $\kappa\in\mathbb{N}_0+l$.
Then, renormalization is non-unique and homogeneous scaling may be broken by logarithmic terms. In the regularized theory, however, $\kappa\not\in\mathbb{Z}$, hence the extension is unique and always homogeneous.

We will show that a similar method works also for the massive case. 

The calculation of principal parts can be performed in terms of integrals in the 
complex plane.
If one wants to iterate the subtraction procedure one has to do these integrations independently. 
This requires the ability to vary the ``dimensions'' of propagators independently.  This is possible in the position 
space formulation. Moreover, also the regularized propagators are distributions on 
Minkowski space. In the momentum space formulation, the dimension  of propagators
has to be chosen for subgraphs, and, in the case of overlapping divergences, these dimensions 
cannot independently be varied.  

A nice feature of dimensional regularization is that many structural properties are respected by the regularization which then are automatically satisfied by the minimally subtracted theory. This holds also for our method and includes in particular Poincar\'{e} invariance, unitarity and the validity of field equations.    
However, our version of dimensional regularization does not preserve gauge invariance, because the propagators are modified and not the integration measure.
A few thoughts, how one may possibly overcome this drawback, are given in the 'Conclusions and Outlook'.

\section{Functional approach and Epstein-Glaser Renormalization}\label{EG}
We restrict ourselves to the theory of a real scalar field. Let $\Ecal(\MM)$ denote the space of smooth functions on the
$d$-dimensional Minkowski space, equipped with its standard Fr\'echet topology, and consider the space of smooth maps $F:\Ecal(\MM)\to\CC$. Let us recall  the definition of smoothness used in infinite dimensional calculus (see, e.g., \cite{Neeb2005}). The derivative of $F$ at $\ph\in\Ecal(\MM)$ in the direction of $\psi$ is defined as
\be\label{de}
F^{(1)}(\ph)[\psi] \doteq \lim_{t\rightarrow 0}\frac{1}{t}\left(F(\ph + t\psi) - F(\ph)\right)\,,
\ee
whenever the limit exists. $F$ is called differentiable at $\ph$ if $F^{(1)}(\ph)[\psi]$ exists for all $\psi\in\Ecal(\MM)$. It is called continuously differentiable in an open neighborhood $U\subset \Ecal(\MM)$ if it is differentiable at all points of $U$ and
$F^{(1)}:U\times \Ecal(\MM)\rightarrow \RR, (\ph,\psi)\mapsto F^{(1)}(\ph)[\psi]$
is a continuous map. It is called a $\mathcal{C}^1$-map if it is continuous and continuously differentiable. Higher derivatives are defined for $\mathcal{C}^n$-maps by 
\be
F^{(n)} (\ph)(\psi_1 , \ldots , \psi_n ) \doteq \lim_{t\rightarrow 0}\frac{1}{t}\big(F^{(n-1)}  (\ph + t\psi_n )(\psi_1 , \ldots, \psi_{n-1} ) -
 F^{(n-1)}(\ph)(\psi_1 , \ldots, \psi_{n-1}) \big)
 \ee
In particular, it means that if $F$ is a smooth functional on $\Ecal(\MM)$, then its $n$-th derivative at the point $\ph\in\Ecal(\MM)$ is a compactly supported distributional density $F^{(n)}(\ph)\in\Ecal'(\MM^n)$. There is a distinguished volume form on $\MM$, so we can use it to construct densities from functions and to provide an embedding of $\Dcal(\MM^n)$
 into $\Ecal'(\MM^n)$. Using the distinguished volume form we can identify derivatives $F^{(n)}(\ph)$ with distributions.

An important property of a functional is its spacetime support. It is defined as a generalization of the distributional support, namely as the set of 
points $x\in \MM$ such that $F$ depends on the field configuration in any neighborhood of $x$.
\begin{align}\label{support}
\supp\, F\doteq\{ & x\in \MM|\forall \text{ neighborhoods }U\text{ of }x\ \exists \ph,\psi\in\Ecal(\MM), \supp\,\psi\subset U 
\\ & \text{ such that }F(\ph+\psi)\not= F(\ph)\}\ .\nonumber
\end{align}
Here we will discuss only compactly supported functionals.
Finally we assume that the wave front set of the distribution $F^{(n)}(\ph)$, considered as a subset of the cotangent bundle
$T^*(\MM^n)=\MM^n\times\MM^n$, does not intersect the set $\MM^n\times(V_+^n\cup V_-^n)$
where $V_{\pm}$ denotes the closed forward and backward light cone, respectively. Functionals fulfilling all the conditions listed above are called microcausal, and the space of such functionals is denoted by $\Fcal$. It contains a subspace $\Floc$, the space of local  functionals, characterized by the additivity condition
\be\label{add}
F(\ph+\psi+\chi)=F(\ph+\psi)-F(\psi)+F(\psi+\chi) \ \text{ if }\supp\phi\cap\supp\chi=\emptyset
\ee   
(as shown in \cite{BDF09} this implies that the derivatives $F^{(n)}(\ph)$ have support on the thin 
diagonal $\Delta=\{(x,\dots,x)|x\in\MM\}$). In addition, the wave front sets of derivatives of local functionals are required to be perpendicular to the tangent bundle of $\Delta$, considered as a subset of the tangent 
space of $\MM^n$. $\Floc$ contains the functionals which occur as local interactions in the EG framework, e.g.
$F(\ph)=\int\, \ph(x)^4f(x)d^dx$ with a test function $f$ with compact support. Finite sums of pointwise products of local functionals form a subalgebra of $\mathcal{F}$, which is called the algebra of multilocal functionals, and we denote it by $\mathcal{F}_{\mathrm{ml}}$. It was shown in \cite{BFLR} that local functionals in the above sense are of the form:
\[
F(\ph)=\int f(x,\ph(x),\partial\ph(x),\dots)d^4x 
\]
where $f$ depends smoothly on $x$ and, for a fixed $\ph$, on finitely many derivatives%
\footnote{In \cite{BFLR} it was shown, with the use of the fundamental theorem of calculus, that $F(\ph)=\int f(j_x^\infty(\ph)) d^4x$, 
where $j_x^\infty(\ph)=(x,\ph(x),\partial\ph(x),\dots)$ is the infinite jet of $\ph$ at $x$. Moreover, the functional derivatives of $F$ 
have compact support on the thin diagonal and are therefore finite derivatives of the $\delta$-distribution in the relative variables 
(i.e.~denoting the latter by $\delta_\Delta$ it holds
$\langle\delta_\Delta,h\rangle=\int d^dx\, h(x,\dots,x)\ ,\,\, h\in\mathcal D(\mathbb M^n)$),
with coefficients which are smooth functions of $x$.} of $\ph$ at $x$.

Dynamics is introduced along the lines of \cite{BDF09} by means of a generalized Lagrangian. It is defined as a map $L:\Dcal(\MM)\rightarrow \Fcal_\loc$ satisfying 
\be\label{L:supp}
\supp(L(f))\subseteq \supp(f)\,,\qquad \forall\, f\in\Dcal(\MM)\,,
\ee
and the additivity condition \eqref{add} with respect to test functions. The action $S(L)$ is defined as an equivalence class of Lagrangians, where two Lagrangians $L_1$, $L_2$ are called equivalent $L_1\sim L_2$  if
\be\label{equ}
\supp (L_{1}-L_{2})(f)\subset\supp\, df\,, 
\ee
in particular two Lagrangians are identified if their densities differ by a total derivative.  
For the free Klein-Gordon field the generalized Lagrangian is given by:
\be\label{free:action}
L_0(f)(\ph)=\frac{1}{2}\int (\partial_\mu\ph\partial^\mu\ph-m^2\ph^2)fd^4x\,.
\ee 
The second functional derivative of $L_0$ induces a 
linear operator, which in our case is the Klein-Gordon operator $P=\Box+m^2$. The free quantized theory is defined by means of deformation quantization of the classical Poisson structure induced by $P$ (see \cite{DF01b,DF01a,BDF09} for details). On Minkowski spacetime one can perform this deformation using the Wightman 2-point function $\Delta_+$, to define a non-commutative product
\be\label{star:prod}
(F\star G)(\ph)\doteq\sum\limits_{n=0}^\infty \frac{\hbar^n}{n!}\left<F^{(n)}(\ph),\left(\Delta_+\right)^{\otimes n}G^{(n)}(\ph)\right>\,,
\ee 
which is interpreted as the operator product of the free theory.

Interaction is introduced in terms of time-ordered products. Let us first consider regular functionals, i.e. such that $F^{(n)}(\ph)\in\Dcal(\MM^n)$ for all $\ph\in\Ecal(\MM)$, $n\in\NN$. We denote the space of such functionals by $\Fcal_\reg$. 
Time-ordered products $\TT_n$, defined on $\Fcal_\reg[[\hbar]]$, are equivalent to the pointwise product 
\be\label{pointwise-product}
m_n(F_1\otimes\dots\otimes F_n)(\ph)=F_1(\ph)\dots F_n(\ph)
\ee 
by the ``heat kernel''
\be
\TT=e^{\frac12 D}
\ee
with $D=\langle\hbar \Delta_F,\frac{\delta^2}{\delta\ph^2}\rangle$ ($\Delta_F$ denotes 
the Feynman propagator), i.e.
\be\label{TT-unren}
\TT_n=\TT\circ m_n\circ (\TT^{-1})^{\otimes n} \ .
\ee
Using Leibniz' rule 
\be
\frac{\delta}{\delta\ph}\circ m_n=m_n\circ(\sum_{i=1}^n\frac{\delta}{\delta\ph_i}) 
\ee
(here an element of the $n$th tensor power of $\Fcal_\reg$ is considered as a functional of $n$ independent field configurations $\ph_1,\dots,\ph_n$) and the notation
\be\label{F1}
D_{ij}=\langle\hbar\Delta_F, \frac{\delta^2}{\delta\ph_i\delta\ph_j}\rangle
\ee
we obtain $\TT_n=m_n\circ T_n$, where
\be\label{F2}
T_n=e^{\sum_{i<j}D_{ij}}=\prod_{i<j}\sum_{l_{ij}=0}^{\infty}\frac{D_{ij}^{l_{ij}}}{l_{ij}!}
\ee
Note that the time-ordered product is commutative and associative. In this paper we will use consequently the calligraphic script (for example $\TT_n$) to denote objects involving the multiplication $m_n$, while roman letters (like $T_n$) are reserved for ``bare'' expressions, where $m_n$ is not applied.

The exponential in the formula \eqref{F2}
might be expanded into a formal power series. This yields the usual graphical description for time-ordered products, since the right hand side of \eqref{F2} may be written as a sum over all graphs $\Gamma$ with vertices $V(\Gamma)=\{1,\dots,n\}$ and $l_{ij}$ lines $e\in E(\Gamma)$ connecting the vertices $i$ and $j$.
We set $l_{ij}=l_{ji}$ for $i>j$ and $l_{ii}=0$ (no tadpoles). If $e$ connects $i$ and $j$ we set $\partial e:=\{i,j\}$.
Then we obtain
\be\label{time:ord}
T_n=\sum_{\Gamma\in \Gcal_n}T_{\Gamma}\,,
\ee
with $\Gcal_n$ the set of all graphs with vertex set $V(\Gamma)=\{1,\dots n\}$ and 
\be\label{GraphDO}
T_{\Gamma}=\frac{1}{\textrm{Sym}(\Gamma)}\langle t_{\Gamma},\delta_{\Gamma}\rangle\,,
\ee
where 
\[\delta_{\Gamma}=\frac{\delta^{2\,|E(\Gamma)|}}{\prod_{i\in V(\Gamma)}\prod_{e:i\in\partial e}\delta\ph_i(x_{e,i})}\]
is a functional differential operator on $\Fcal_\reg^{\otimes n}$,
\be\label{SGamma}
t_{\Gamma}=\prod_{e\in E(\Gamma)}\hbar\Delta_F(x_{e,i},i\in\partial e)
\ee
and the, so called, symmetry factor $\textrm{Sym}$ is the number of possible permutations of lines joining 
the same two vertices, $\textrm{Sym}(\Gamma)=\prod_{i<j}l_{ij}!$. We point out that in our approach the Feynman 
graphs are not fundamental objects of the theory, instead they are a bijective graphical description of the 
terms appearing in the exponential function \eqref{F2}, from which one can read off
the analytic expression (including the numerical prefactor) of each term/graph.
{\small \begin{example}[Graph expansion]
Regard three  functionals $F,G,H\in\Fcal_\reg$. Their time-ordered product is given by
\be
\TT_3(F\otimes G\otimes H)(\varphi)=
\sum_{k=0}^\infty\frac1{k!}\left(D_{12}+D_{23}+D_{13}\right)^k F(\varphi_1) G(\varphi_2)H(\varphi_3)
\bigg|_{\varphi_1=\varphi_2=\varphi_3=\varphi}
\ee
Applying the multinomial theorem and inserting the definition for $D_{ij}$ gives
\begin{align}
\TT_3(F\otimes G\otimes H)(\varphi)&\nonumber\\
&\hspace{-30mm}=\sum_{k=0}^\infty\sum_{k_1+k_2+k_3=k}\frac{\hbar^k}{k_1!\,k_2!\,k_3!}\label{gexpansion}\\
&\hspace{-30mm}\qquad\left\langle (\Delta_F^{12})^{\otimes k_1}(\Delta_F^{23})^{\otimes k_2}(\Delta_F^{13})^{\otimes k_3}, F^{(k_1+k_3)}(\varphi_1) G^{(k_1+k_2)}(\varphi_2)H^{(k_2+k_3)}(\varphi_3) \right\rangle\bigg|_{\varphi_1=\varphi_2=\varphi_3=\varphi}\nonumber
\end{align}
The terms in this expression are identified as the usual Feynman graphs in the following way: we assign to lines Feynman propagators and functional derivatives derivatives of given functionals to vertices. Formula \eqref{gexpansion} can now be represented as:
\begin{align}
&=\FGH+\hbar\left(\FoneGHF+\FGoneHF+\FGHoneF\right)\\
&\phantom{=\FGH}+\hbar^2\left[\FoneGHoneF+\FoneGoneHF+\FGoneHoneF +\frac12\left(\FtwoGHF+\FGtwoHF+\FGHtwoF\right)\right]
+\mathcal{O}(\hbar^3)\nonumber
\end{align}
In case the functionals are polynomial in the field and its derivatives, only a finite number of functional derivatives are non-vanishing. Only those graphs remain where the valence at the vertices is bounded by the degree of the associated polynomial functionals. 
\end{example}}
For regular functionals $F\in\mathcal{F}_{\mathrm{reg}}$, the contraction $\langle t_{\Gamma},\delta_{\Gamma}\rangle$ is  well-defined since $t_{\Gamma}$ is applied to a test function in the corresponding dual space. For local functionals, however, the functional derivatives are of the form
\be\label{dF}
F^{(k)}(\ph)(x_1,\dots,x_k)=\sum_\bet f^{(k)}_\bet(x_{\text{cms}})\partial^{\bet}\de(x_{\text{rel}})
\ee
where $\beta\in\NN_0^{d(k-1)}$, test functions 
$f^{(k)}_\bet(x)\equiv f^{(k)}_\bet(\ph)(x)\in\Dcal(\MM)$ are functions of the center of 
mass coordinate $x_{\text{cms}}=(x_1+\dots+x_k)/k$
which depend on $\ph$, and $x_{\text{rel}}=(x_1-x_{\text{cms}},\dots, x_k-x_{\text{cms}})$ denotes the relative coordinates. 
Hence, the functional differential operator $\delta_\Gamma$ applied to $\mathcal{F}_{\mathrm{loc}}^{\otimes n}$ yields, at any $n$-tuple of field configurations $(\varphi_1,\dots,\varphi_n)$, a compactly supported distribution in the variables $x_{e,i},i\in\partial e, e\in E(\Gamma)$ with support on the partial diagonal $\Delta_{\Gamma}=\{x_{e,i}=x_{f,i},i\in\partial e\cap\partial f, e,f\in E(\Gamma)\}$ with a wavefront set perpendicular to $T\Delta_{\Gamma}$. Such a distribution can uniquely be written as a finite sum 
\be
f=\sum f_\beta\otimes\partial^\beta\delta\label{factorisation}
\ee
where $f_{\beta}\in\mathcal{D}(\Delta_\Gamma)$ and where $\partial^\beta\delta$ (with a multi-index $\beta$) is a partial derivative of the $\delta$-distribution on the orthogonal complement of $\Delta_\Gamma$. A concrete coordinatization of $\Delta_\Gamma$ can be given by the center of mass coordinates introduced above and the coordinates on the orthogonal complement can be chosen as the relative coordinates. To obtain \eqref{factorisation}, one has to write all partial derivatives $\partial_{x_{e,i}}$ in terms of partial derivatives in $x_{\text{cms}}$ and $x_{\text{rel}}$ coordinates. The former are applied on the $\ph$-dependent test function and produce $f_\beta$ and the latter are applied on the Dirac $\delta$ distribution.

Let $Y_{\Gamma}$ denote the vector space spanned by the distributions $\partial^\beta\delta$. $Y_\Gamma$ is graded by the order of the derivatives. The space of $Y_\Gamma$-valued test functions on $\Delta_\Gamma$ is denoted by $\mathcal{D}(\Delta_\Gamma,Y_\Gamma)$. One now has to define the action of the distribution $t_\Gamma$ as a linear functional on $\mathcal{D}(\Delta_\Gamma,Y_\Gamma)$,
\be\label{time ordered functions}
\langle t_\Gamma,f\rangle=\sum \langle t_{\Gamma}^{\beta},f_{\beta}\rangle
\ee
with numerical distributions $t_{\Gamma}^{\beta}\in\mathcal{D}'(\Delta_{\Gamma})$.
{\small\begin{example} Let $F_1=\int dx \,g(x)\,(\ph^2(\partial\ph)^2)(x)\ ,\,\,
F_2=\int dx \,h(x)\,\ph^3(x)\ ,\,\,g,h\in\Dcal(\MM)$, 
$t_\Gamma=\hbar^2\,\Delta_F(x_{11}-x_{12})\Delta_F(x_{21}-x_{22})$ and 
$\delta_\Gamma=\tfrac{\delta^4}{\delta\ph_1(x_{11})\delta\ph_1(x_{21})\delta\ph_2(x_{12})\delta\ph_2(x_{22})}$.
Then,
\begin{align*}
f=&\delta_\Gamma(F_1(\ph_1)F_2(\ph_2))=\int dx_1 \,g(x_1)\int dx_2 \,h(x_2)\,6\,\ph_2(x_2)\,\delta(x_{12}-x_2)\delta(x_{22}-x_2)\\
&\Bigl(2\,(\partial\ph_1)^2(x_1)\,\delta(x_{11}-x_1)\delta(x_{21}-x_1)+2\,\ph_1^2(x_1)\,
\partial\delta(x_{11}-x_1)\partial\delta(x_{21}-x_1)\\
&-4\,(\ph_1\partial\ph_1)(x_1)\bigl(\delta(x_{11}-x_1)\partial\delta(x_{21}-x_1)+(x_{11}\leftrightarrow x_{21}\bigr)\Bigr)
\end{align*}
and with that we obtain
\begin{align*}
\langle t_\Gamma,&f\rangle=\hbar^2\,\int dx_1 \,g(x_1)\int dx_2 \,h(x_2)\,\Bigl(
12\,[(\Delta_F(x_1-x_2))^2]^\mathbf{\cdot}\,(\partial\ph_1)^2(x_1)\,\ph_2(x_2)+\\
&12\,[(\partial\Delta_F(x_1-x_2))^2]^\mathbf{\cdot}\,\ph_1^2(x_1)\,\ph_2(x_2)+
48\,[\Delta_F(x_1-x_2)\partial\Delta_F(x_1-x_2)]^\mathbf{\cdot}\,\ph_1(x_1)\partial\ph_1(x_1)\,\ph_2(x_2)\Bigr)\ .
\end{align*}
Hence, modulo constant prefactor
the appearing numerical distributions $t^\beta_\Gamma$ are 
the extensions (denoted by $[\cdots]^\mathbf{\cdot}$)
of $(\Delta_F(x_1-x_2))^2$, $(\partial\Delta_F(x_1-x_2))^2$
and $\Delta_F(x_1-x_2)\partial\Delta_F(x_1-x_2)$, respectively,
from $\Dcal'(\MM^2\setminus\Delta_2)$ to $\Dcal'(\MM^2)$ (where 
$\Delta_2$ is the diagonal $x_1=x_2$).
\end{example}}
The method of Epstein and Glaser proceeds by induction. One proves that if 
$t_{\Gamma'}$ is known for all graphs $\Gamma'$ with less vertices than $\Gamma$, then $t_\Gamma$ can be uniquely defined for all disconnected, all connected one particle reducible and all one particle irreducible one vertex reducible graphs. For the remaining graphs (which we call EG-irreducible) one can define it uniquely on all distributions $f\in\mathcal{D}(\Delta_{\Gamma},Y_\Gamma)$ of the form above where $f_\beta$ vanishes together with all its derivatives of order  $\leq\omega_\Gamma+|\beta|$ on the thin diagonal of $\Delta_\Gamma$. Here
\[\omega_\Gamma=(d-2)|E(\Gamma)|-d(|V(\Gamma)|-1)\]
is the degree of divergence of the graph $\Gamma$.
We  denote this subspace by $\mathcal{D}_{\omega_{\Gamma}}(\Delta_\Gamma,Y_\Gamma)$ .

Renormalization then amounts to project a generic $f$ to this subspace by a translation invariant projection $W_{\Gamma}:\mathcal{D}(\Delta_\Gamma,Y_\Gamma)\to\mathcal{D}_{\omega_\Gamma}(\Delta_\Gamma,Y_\Gamma)$. Different renormalizations differ by different choices of the projections $W_\Gamma$.

The ambiguity is best described in terms of the Main Theorem of Renormalization \cite{SP82,Pin01,DF04,BDF09}. 
Let the formal S-matrix be defined as the generating functional of time-ordered products, formally given by
\be
\Scal=\TT\circ\exp\circ \TT^{-1}\ ,
\ee  
i.e., its $n$th derivative at zero, $\TT_n\equiv \Scal^{(n)}(0)$, as a linear map 
$\TT_n:\Floc[[\hbar]]^{\otimes n}\to\Fcal[[\hbar]]$ is the (renormalized) time-ordered $n$-fold product. 
Given two $S$-matrices $\Scal$ and $\hat{\Scal}$ fulfilling the Epstein-Glaser axioms, there exists
 a unique analytic map $\Zcal:\Floc[[\hbar]]\to\Floc[[\hbar]]$ with $\Zcal(0)=0$ such that
\be
\hat{\Scal}=\Scal\circ \Zcal\ .
\ee   
To first order, this relation gives\footnote{Similarly to $\Scal^{(n)}$ we write $\Zcal^{(n)}$ for $\Zcal^{(n)}(0)$.} 
$\Zcal^{(1)}=\mathrm{id}$. 
The maps $\Zcal$ relating different $S$-matrices in this way form the renormalization group  
$\mathscr{R}$ in the sense of St\"uckelberg and Petermann. It is a subset of the group of formal diffeomorphisms 
(tangent to the identity) on the space of interactions. A direct definition 
of $\mathscr{R}$ by the properties of the maps $\Zcal$ 
is given in \cite{DF04,BDF09}: $\mathscr{R}$ has the structure of an affine space,
\be\label{R-affine}
\mathscr{R}=\mathrm{id}+\hbar\,\mathscr{V}[[\hbar]]\ ,
\ee
where $\mathscr{V}[[\hbar]]$ is a {\it vector space} of formal power series $\mathcal{V}=\sum_{n=0}^\infty \mathcal{V}_n\,\hbar^n$, which are analytic maps
$\mathcal{V}:\Floc[[\hbar]]\to\Floc[[\hbar]$, the main defining properties of elements of $\mathscr{V}$ are $\mathcal{V}(0)=0\ ,\,\,\mathcal{V}^{(1)}(0)=0$ and {\it locality},
\be\label{R-locality}
\mathcal{V}(F+G+H)=\mathcal{V}(F+G)-\mathcal{V}(G)+\mathcal{V}(G+H)\quad\text{if}\quad \supp\,F\cap\supp\,H=\emptyset\ ,\quad\forall\ \mathcal{V}\in\mathscr{V}\ .
\ee
To show that $\mathscr{R}$ is indeed a group, one needs additionally the property (proved in \cite{DF04}) that,
given an S-matrix $\Scal$ and $\Zcal\in  \mathscr{R}$, the composition $\hat{\Scal}:=\Scal\circ\Zcal$ satisfies 
also the Epstein-Glaser axioms.

One of the great virtues of the Epstein-Glaser approach is that it does not involve any divergences, and that it is explicitly independent of any regularization prescription. It can therefore be used for an \textit{a priori} definition of the problem of the perturbative construction of quantum field theory which then is solved by the method of renormalization. In other schemes usually only an \textit{a posteriori} definition is possible, and the independence of the construction from the chosen method relies on a comparison with other methods.

We have just outlined how to define the $n$-fold time-ordered products (i.e. multilinear maps $\TT_n$) by the procedure of Epstein and Glaser. An interesting question is whether the renormalized time-ordered product defined by such a sequence of multilinear maps can be understood as an iterated binary product on a suitable domain. 
Recently it was proven in  \cite{FRb} that this is indeed the case. The crucial observation is that 
pointwise multiplication of local functionals is injective. More precisely, let $\Fcal_0$ be the set of local functionals vanishing at some distinguished field configuration (say $\ph=0$). Iterated multiplication $m$ is then a linear map from the symmetric Fock space over $\Fcal_0$ onto $\Fcal_\ml$. Then, the following assertion holds:
\begin{prop}\label{beta}The multiplication $m:S^\bullet\Fcal_0\to\Fcal_{\ml}$ is bijective (where $S^k$ denotes the symmetrised tensor product of vector spaces).
\end{prop} 
Let $\beta=m^{-1}$. We now define the renormalized time-ordering operator on the space of multilocal functionals $\Fcal_{\ml}$ by
\be\label{T-ren}
\TT_\mathrm{ren}:=(\bigoplus_n  \TT_n)\circ\beta 
\ee  
This operator is a formal power series in $\hbar$ starting with the identity, hence it is injective.
The corresponding binary product $\T$ is now defined on the image of $\TT_\mathrm{ren}$ by
\be\label{Tprod}
F\T G\doteq \TT_\mathrm{ren}(\TT_\mathrm{ren}^{-1}F\cdot \TT_\mathrm{ren}^{-1}G)\,,
\ee
This product is equivalent to the pointwise product and, hence, it 
is in particular associative and commutative. Moreover, the $n$-fold iteration
of the binary product $\T$ applied to local functionals  
coincides with the linear 
map $\TT_n$ defined by the Epstein-Glaser procedure.

We may now use the St\"uckelberg-Petermann group $\mathscr{R}$ in order to establish a relation between the renormalized and the 
regularized S-matrix. Let $\Delta_F^{\Lambda}$ be a regularized Feynman propagator, and let the upper index $\Lambda$ 
indicate that in the formal construction the regularized propagator was used. A regularization should satisfy the 
condition that all the expressions for time-ordered products become meaningful for local functionals 
(still in the sense of formal power series in $\hbar$), 
and that the regularized propagators converge in the sense of the H\"ormander topology on distributions on 
$\RR^d\setminus\{0\}$ with the appropriate wave front sets and microlocal scaling degrees. The former condition is surely 
satisfied if $\Delta_F^{\Lambda}$ is a smooth function of rapid decrease. 
\section{Analytic regularization, Minimal Subtraction and Forest Formula}\label{sec:forest-MS}
Let $\Scal^{\Lambda}$ be the regularized S-matrix constructed from $\Delta_F^{\Lambda}$, more precisely $\Scal^{\Lambda}$ is the formal power series
$$
\Scal^{\Lambda}:=1+\mathrm{id}+\sum_{n=2}^\infty\frac{1}{n!}\,m_n\circ \exp\sum_{1\leq i<j\leq n}D^\Lambda_{ij}\ ,\quad
D^\Lambda_{ij}:=\langle\hbar\Delta^\Lambda_F, \frac{\delta^2}{\delta\ph_i\delta\ph_j}\rangle\ .
$$
To relate the construction of the S-matrix $\Scal$ of Epstein-Glaser to the method of divergent counter terms, we search a
family of renormalization group elements $Z^{\Lambda}\in\mathscr{R}$ such that
\be\label{counter terms}
\Scal=\lim_{\Lambda\to\Lambda_0}\Scal^{\Lambda}\circ \Zcal^{\Lambda} \ .
\ee
If $\Scal$ is given, then such a family $(\Zcal^{\Lambda})$ exists and this family  
is uniquely determined up to a sequence which converges to the identity (see the Appendix \ref{app:regularization} and \cite[Sect.~5.2]{BDF09}). 

A special role is played by analytic regularization schemes where $\Scal^{\Lambda}$ is a meromorphic function of 
$\Lambda\in\CC$ with a pole at the limit point $\Lambda_0$. In these cases there exists a distinguished 
choice $\Scal_{\MS}$ of the S-matrix (\emph{minimal subtraction}) and the corresponding family of renormalization group elements $\Zcal_{\MS}^{\Lambda}\in\mathscr{R}$. To construct these objects we start with the family $(\Zcal^{\Lambda})$ of meromorphic functions and we prove the so called Birkhoff decomposition (see \cite{Connes1999,CK00}, where such notion was first introduced in the context of renormalization):
\[
\Zcal^{\Lambda}=\Zcal^{\Lambda}_{\MS}\circ \Zcal_f^{\La}\,,
\] 
where $\Zcal_f^{\La}$ is regular and $\Zcal^{\Lambda}_{\MS}$ corresponds to subtracting the principal part. We prove this by induction. Let us consider the n-th functional derivative ${\Zcal^{\Lambda}}^{(n)}$ and we assume that for $k<n$ we have already constructed ${\Zcal_f^{\La}}^{(k)}$ and   ${\Zcal^{\Lambda}_{\MS}}^{(k)}$ in such a way that, for $l<n$, the chain rule (Fa\`a di Bruno formula) 
$${\Zcal^{\Lambda}}^{(l)}=\sum_{P\in\mathrm{Part}(\{1,\dots,l\})}(\Zcal^{\Lambda}_{\MS})^{(|P|)}\circ\bigotimes_{I\in P} (\Zcal^{\Lambda}_f)^{(|I|)}\quad\textrm{holds.}$$
The function
\[
{\Zcal^{\Lambda}}^{(n)}-\sum_{P\in\mathrm{Part}(\{1,\dots,n\})\atop 1<|P|<n
}(\Zcal^{\Lambda}_{\MS})^{(|P|)}\circ\bigotimes_{I\in P} (\Zcal^{\Lambda}_f)^{(|I|)}
\]
is meromorphic, so we can decompose it as a sum of the principal part, which we call ${\Zcal^{\Lambda}_{\MS}}^{(n)}$, and the rest term ${\Zcal_f^{\La}}^{(n)}$. This way we construct the $n$-th order derivative of  $\Zcal^{\Lambda}_{\MS}$ and we can proceed by induction. Using that $\Zcal_\Lambda\in\mathscr{R}$ one easily sees that $\Zcal^{\Lambda}_{\MS}$ satisfies \eqref{R-locality} for $G=0$. This implies the general case since all the quantities are formal power series (see Appendix B of \cite{BDF09}). It follows that $\Zcal^{\Lambda}_{\MS}$ and $\Zcal^\Lambda_f$ obtained by the above construction are elements of the St\"uckelberg-Petermann group $\mathscr{R}$.

By construction, $\Zcal_f^{\La}$ has a well defined limit as $\La$ approaches $\La_0$ and, since it is invertible, we can define $\Scal_{\MS}^\La:= \Scal^\La\circ \Zcal^\La\circ{\Zcal_f^{\La}}^{-1}=\Scal^\La\circ \Zcal_{\MS}^\La$ and this expression also has a well defined limit,
\[
\Scal_{\MS}:= \lim\limits_{\La\rightarrow \La_0}\Scal_{\MS}^\La\,.
\]
It can be expressed as $\Scal_{\MS}=\Scal\circ \Zcal_f^{-1}$, where $\Zcal_f:=\lim\limits_{\Lambda\rightarrow \La_0}\Zcal^\La_f$ and, because $\Zcal_f$ is an element of $\Rcal$, $\Scal_{\MS}$ is an S-matrix fulfilling the Epstein-Glaser axioms. 
It is the generating functional for minimally subtracted time-ordered products ($MS$ scheme). 

We will now derive a useful recursive formula for $\Zcal^{\Lambda}_{\MS}$. Consider the functional derivative
\be\label{FaadiBruno}
(\Scal_{\MS}^\La)^{(n)}=(\Zcal_{\MS}^\La)^{(n)}+\sum_{P\in\mathrm{Part}(\{1,\dots,n\})\atop 1<|P|}(\Scal^{\Lambda})^{|P|}\left(\bigotimes_{I\in P} (\Zcal_{\MS}^{\Lambda})^{|I|}\right)\,.
\ee
Since $(\Scal_{\MS}^\La)^{(n)}$ converges for $\La\rightarrow \La_0$, the principal parts of the summands above must add up to 0, so we obtain a recursion
\be\label{recursion}
(\Zcal^{\Lambda}_{\MS})^{(n)}=-\pp 
\sum_{|P|>1}(\Scal^{\Lambda})^{|P|}\left(\bigotimes_{I\in P} (\Zcal^{\Lambda}_{\MS})^{|I|}\right)
\ee
together with $(\Zcal^{\Lambda}_{\MS})^{(1)}=\mathrm{id}$.

One may now solve the recursive definition of the minimally subtracted S-matrix in terms of an 
analogue of Zimmermann's Forest Formula.
We define an Epstein-Glaser forest $F=\{I_1,...,I_k\}\in \mathfrak{F}_{\bar n}$ 
to be a set of subsets $I_j\subset\bar n:=\{1,\dots,n\}$ 
which contain at least two elements, $|I_j|\geq 2$, and which satisfy
\be
I_i\cap I_j=\emptyset\quad\text{or}\quad I_i\subset I_j \quad\text{or}\quad 
I_j\subset I_i\quad\forall i<j\ .\nonumber
\ee
The empty set of subsets is referred to as the empty forest. We assume that we can vary the regularization 
parameters $\Lambda_{ij}$ independently for every pair of indices $1\le i<j \le n$ such that the regularized 
time-ordered product is a well defined meromorphic function in all these variables \footnote{For the definition 
of the meromorphic function of several variables, see for example \cite{Lang}.}. 
More precisely, we 
assume that in the graph expansion every distribution $t_{\Gamma}^{\beta}$ (see \eqref{time ordered functions}) 
is, after evaluation on a test function, an analytic function on a suitable domain which can be extended to a meromorphic function on a domain containing $\{\Lambda_{ij}=\Lambda_0, 1\le i<j\le n\}$.  
Now, given a forest 
$F\in \mathfrak{F}_{\bar n}$, we reduce the number of parameters $\Lambda_{ij}$ as follows:
for each $I\in F$ we set $\Lambda_{ij}=\Lambda_I$ for all $i,j\in I$.
Let $R_I$ be (-1) times the projection onto the principal part with respect to the variable $\Lambda_I$.
We then obtain the EG Forest Formula 
\begin{thm}\label{forest}
\be\label{EGforest}
\Scal_{\MS}^{(n)}=\lim_{{\bf \Lambda}\to {\bf \Lambda}_0}
m_n\circ \Bigl(\sum_{F\in\mathfrak{F}_{\bar n}}\prod_{I\in F}R_I\Bigr)
\exp{\sum_{1\leq i<j\leq n}D^{\bf \Lambda}_{ij}}
\quad\quad({\bf \Lambda}\equiv(\Lambda_{ij})_{1\leq i<j\leq n})\ ,
\ee   
where as in Zimmermann's formula $R_I$ has to be applied before $R_J$ if $I\subset J$. The expression $\exp{\sum_{1\leq i<j\leq n}D^{\bf \Lambda}_{ij}}$ has to be understood as a meromorphic function obtained, term by term, by analytic continuation from the region where it exists due to sufficient regularity of the modified Feynman propagators.
\end{thm}
\begin{proof} We omit the index ${\bf \Lambda}$ belonging to each $\Zcal$ and each differential operator $D$,
since it is unessential to the proof. Let us define a full forest as a forest containing the set $\{1,\dots,n\}$, 
and let $\mathfrak{F}_{\bar n}^{\text{full}}$ denote the set of full forests.
We set $\Zcal^{(1)}:=\mathrm{id}$ and
\be\label{counter}
\Zcal^{(n)}:=m_n\circ \Bigl(\sum_{F\in \mathfrak{F}_{\bar n}^{\text{full}}}\prod_{I\in F}R_I\Bigr)\exp{\sum_{i<j}D_{ij}}\ ,
\ee
and we verify that it satisfies the recursion relation \eqref{recursion}, 
i.e. $\Zcal^{(n)}=\Zcal_{\MS}^{(n)}$.
In order to include the case $n=1$ into the formula \eqref{counter} we define 
$\mathfrak{F}_{\{1\}}^{\mathrm{full}}=\{\{1\}\}$ and we adopt the convention that $R_I=\mathrm{id}$ if $I$ contains only one element.

We proceed by induction. For $n=2$ the only full forest is $\{\{1,2\}\}$, hence
\be
\Zcal^{(2)}=m_2\circ R_{\{1,2\}}\exp{D_{12}} 
\ee
in agreement with the definition of minimal subtraction \eqref{recursion} in second order.
For $k>2$ we now assume that $\Zcal^{(n)}=\Zcal_{\MS}^{(n)}$ for all $n<k$.
Let $F\in\mathfrak{F}_{\bar k}^{\mathrm{full}}$ be a full forest. Then there exists a partition $P$ of $\bar k$ such that
\be\label{ff}
F=\{\{1,\dots,k\}\}\cup\bigcup_{L\in P}F_L
\ee 
with full forests $F_L\in\mathfrak{F}_{L}^{\mathrm{full}}$, and $P$ and $F_L$ are uniquely determined. Vice versa, given a partition $P$ and full forests $F_L$, $L\in P$,
equation \eqref{ff} defines a full forest.  
Using Leibniz' rule and the associativity of the pointwise product, we find for a partition $P$ of ${\bar k}$
\be\label{as}
m_{|P|}\circ\exp{\sum_{I<J\in P}D_{IJ}}
\left(\bigotimes_{L\in P}m_{|L|}\circ
\exp{\sum_{i<j\in L}D_{ij}}\right)=
m_k\circ\exp{\sum_{i<j}D_{ij}}
\ee
 where $I<J\in P$ means that $I,J\in P$ and the smallest element of $I$ is smaller than the smallest element of $J$.
This formula, applied to local functionals $F_1,\dots,F_k$, holds on a suitable domain in the deformation parameters and, by analytic continuation,  everywhere as an identity for meromorphic functions.  

We now insert the decomposition of a full forest \eqref{ff} into equation \eqref{counter} and use the identity \eqref{as}. 
We find
\be\label{Zkmin}
\Zcal^{(k)}=R_{\{1,\dots,k\}}\sum_{|P|>1}m_{|P|}\circ\exp{\sum_{I<J\in P}D_{IJ}}
\left(\bigotimes_{L\in P}m_{|L|}\circ\Bigl(\sum_{F\in \mathfrak{F}_L^{\text{full}}}\prod_{M\in F}R_M
\Bigr)\exp{\sum_{i<j\in L}D_{ij}}\right)
\ee
where we used the fact that the operation $R_M$  of taking the principal part involves only the variables $\Lambda_{ij}$ with $i,j\in M$.
But \eqref{Zkmin} is just the recursion relation which defines the minimal subtraction.  
 
In the last step we insert the formula \eqref{counter} into the Fa\`a di Bruno formula \eqref{FaadiBruno} and 
repeat the calculation  
with the modifications that $R_{\{1,\dots,k\}}$ and the restriction $|P|>1$
are omitted. As a result we obtain \eqref{EGforest}.
\end{proof}

{\small \begin{example}[Forest Formula for a particular graph]
The Forest Formula (\ref{EGforest}) can be broken down to the renormalization of individual graphs. As an example we want to regard the following overlapping divergence in $\varphi^4$-theory in 4 dimensions,
\be
G=\Catseye.
\ee
It is a contribution to the selfenergy in fourth order of causal perturbation theory. Introducing an irrelevant numbering of vertices the corresponding differential operator for the graph is:
\be
m_4\circ \left[D_{12}D_{13}D_{14}D_{23}^2D_{24}D_{34}\right]
\ee
Next we write down the basic subsets, from which the EG forests for any four point graph are built:
\be
\{1,2\},\{1,3\},\{1,4\},\{2,3\},\{2,4\},\{3,4\},
\{1,2,3\},\{1,3,4\},\{1,2,4\},\{2,3,4\},
\{1,2,3,4\}.
\ee
Only some of these subsets correspond to divergent subgraphs of $G$:
\be\label{div-graphs}
\{2,3\},
\{1,2,3\},\{2,3,4\},
\{1,2,3,4\}.
\ee
Thus, the relevant forests are
\be\nonumber
\begin{array}{cccccc}
\{\},    &\{23\},     &\{123\},     &\{234\},     &\{23,123\},     &\{23,234\},\\
\{1234\},&\{23,1234\},&\{123,1234\},&\{234,1234\},&\{23,123,1234\},&\{23,234,1234\},
\end{array}
\ee
where we wrote the normal forests in the first, and the corresponding full forests in the second line. The Forest Formula thus yields
\begin{align*}
G_\MS
&= m_4\circ [1+R_{23}+R_{123}+R_{234}+R_{123}R_{23}+R_{234}R_{23}+R_{1234}+R_{1234}R_{23}\\
&\phantom{==m_4\circ}+R_{1234}R_{123}+R_{1234}R_{234}+R_{1234}R_{123}R_{23}+R_{1234}R_{234}R_{23}]\;G\\
&= m_4\circ(1+R_{1234})(1+R_{123}+R_{234})(1+R_{23})\left[D_{12}D_{13}D_{14}D_{23}^2D_{24}D_{34}\right].
\end{align*}
The second equality, i.e.~that $\sum_{F}\prod_{I\in F}R_I$ can be written as $\prod(1+\sum R_I)$, is a peculiarity of this example,
which is due to the fact that each divergent subgraph \eqref{div-graphs} is a subgraph of all divergent (sub)graphs
of higher orders.
\end{example}}
Simpler examples, for which we will compute the projections $R_I$, are \ref{exp:triangle1}
and \ref{thm:doubletriangle}.

\section{Dimensional regularization in position space}\label{sec:dimreg}
To perform the analytic regularization required for minimal subtraction and the EG Forest Formula we have to find a distribution valued analytic function $\zeta\mapsto \Delta_F^{\zeta}$ with the following properties: for $\zeta =0$ the distribution should coincide with the Feynman propagator, the wave front set of $\Delta_F^{\zeta}$ should always be contained in the wave front set of $\Delta_F$, 
and the scaling degree \eqref{sd} of $\Delta_F^{\zeta}$, modulo smooth functions, should tend to $-\infty$ as the real part of $\zeta$ approaches infinity. Under these conditions each term in the graph expansion of the unrenormalized S-matrix is a well defined analytic function for suitable values of $\zeta$. One then has to perform an analytic extension to a meromorphic function with a pole at $\zeta =0$. For the Forest Formula we also require that this extension can be done individually for every propagator associated to a pair of vertices. Under these conditions minimal subtraction is well defined and the EG Forest Formula yields a closed expression. Depending on the choice of the analytic regularization the minimally subtracted S-matrix may automatically satisfy further conditions (e.g. Lorentz invariance, unitarity etc.).

In the following we concentrate on dimensional regularization. According to Bollini and Giambiaggi \cite{BG96} dimensional regularization in position space essentially amounts, in the massive case,  to a change in the index of the Bessel function appearing in the formula for the Feynman propagator. We have in $d$ dimensions
\[\Delta_F(x)=\lim_{\epsilon\searrow 0}\wcal^d(x^2-i\epsilon)\]
where
\[\wcal^d(z^2)=(2\pi)^{-\frac{d}{2}}m^{\frac{d}{2}-1}\sqrt{-z^2}^{1-\frac{d}{2}}K_{\frac{d}{2}-1}(m \sqrt{-z^2})\] 
with the modified Bessel function of the second kind $K_\nu$ with index $\nu=d/2-1$. The massless case is obtained as the limit $m\downarrow 0$.  The dimensionally regularized Feynman propagator is obtained by replacing $d$ by $d-2\zeta$ and, in order to keep the mass dimension constant, by multiplication with a factor $\mu^{2\zeta}$ with a mass parameter $\mu> 0$,
\be\label{regFeynprop}
\Delta^{\zeta}_F(x)=\mu^{2\zeta}\lim_{\epsilon\searrow 0}\wcal^{d-2\zeta}(x^2-i\epsilon)\ .
\ee
Since $x^{-\nu} K_\nu(x)$ is analytic in the cut plane $\CC\setminus\RR^0_-$ 
irrespective of the value of $\nu\in\CC$, the function $w^{d-2\zeta}$ is 
analytic on the complement of the positive real axis.   
Therefore $\Delta^{\zeta}_F(x)$ is, for all values of $\zeta$, the boundary value of an 
analytic function on the complexified Minkowski space with domain 
$\{x+iy|(x+iy)^2\not\in\mathbb{R}_-\}$, as for timelike $x$
the imaginary part $y$ approaches zero from the backward light cone if $x^0>0$ and 
from the forward light cone if $x^0<0$.
By \cite[Thm.~8.1.6]{Hoer03} the analyticity domain of a function determines the wave front 
set of its boundary values, hence $\WF(\De_F^\zeta)\subset\WF(\Delta_F)$.
Moreover, the scaling degree of $\Delta_F^\zeta$ is $\mathrm{max}\{d-2-2\Re\,\zeta,\,0\}$, 
as may be seen from the behavior of the Bessel function at the origin.

We now define the terms of the regularized S-matrix for suitable values of $\zeta$. Using the information on the scaling degrees as well as on the wave front sets we can proceed in the usual way for the Epstein-Glaser induction and  find, at every step, a unique expression which has the correct wave front sets and analyticity properties. It remains to construct the analytic continuations.

For this purpose we use the fact that $\Delta_F^\zeta$ can be written as a sum of homogeneous distributions and a rest term with sufficiently small scaling degree where every term is analytic in $\zeta$.  Every contribution in the expansion is then a product of homogeneous distributions and sufficiently well behaved functions. But for non-integer degree of homogeneity, homogeneous distributions have unique homogeneous extensions. This extension is analytic in $\zeta$, hence the whole expression is analytic for non-integer values of $\zeta$. In case one uses different values of $\zeta$ the degree of homogeneity differs from its value for $\zeta_{ij}=0$ by
a certain linear combination $\sum_{i<j}l_{ij}\zeta_{ij}$ of the $\zeta$-variables, and one may find poles at points where these degrees are integers.  

The choice of dimensional regularization has further nice properties. First of all, since the regularized propagators are Lorentz invariant, one automatically obtains a Lorentz invariant S-matrix. Moreover, due to $\overline{\Delta_F^\zeta(x)}=\Delta_{AF}^{\overline{\zeta}}(x)$ (where $\Delta_{AF}^{\zeta}$ is the regularized anti-Feynman 
propagator which is obtained from \eqref{regFeynprop} by reversing the sign of $\epsilon$),
the S-matrix is unitary. Finally, the condition that the field equation holds, is due to the fact that $\Delta_F^\zeta$ is analytic at $\zeta=0$.  

Before we enter in the details of the regularization of the two point functions 
(sect.~\ref{sec:DimRegHadamard}) and of the construction of a pertinent sequence $(\TT^\bzeta_n)$
of regularized time-ordered products 
(sect.~\ref{sec:dimregS}), we study regularization and minimal subtraction on the 
level of numerical distributions in the Epstein-Glaser framework.

\subsection{Regularization of numerical distributions}\label{MSnumerical}

In a translation invariant framework, perturbative renormalization can be understood in $x$-space as the extension
of distributions $t\in\Dcal'(\RR^n\setminus\{0\})$ to distributions $\dot t\in\Dcal'(\RR^n)$ \cite{Stora1993}. 
This mathematical problem is 
treated, for example, in \cite{BF00,DF04}. 
The existence and uniqueness of extensions $\dot t$ can be answered in terms of
Steinmann's scaling degree \cite{Ste71},
\be\label{sd}
\mathrm{sd}(t):=\inf\{\omega\in\RR\,|\,\lim_{\rho\downarrow 0}\rho^\omega\,t(\rho x)=0\}\ ,\quad
t\in\Dcal'(\RR^n)\quad\text{or}\quad  t\in\Dcal'(\RR^n\setminus\{0\})\ .
\ee
For the complete existence and uniqueness theorem we refer to \cite{BF00}, we only mention the following partial result. 
\begin{thm}[{\cite{Ste71,BF00}}]\label{thm:Extension-sd} 
For $\lambda\in\RR$ let
\be
\Dcal_\lambda(\RR^n):=\{f\in\Dcal(\RR^n)\,|\,(\partial^\alpha f)(0)=0\,\,\,\forall |\alpha|\leq\lambda\}
\ee
(in particular $\Dcal_\lambda(\RR^n)=\Dcal(\RR^n)$ if $\lambda <0$)
and let $\Dcal'_\lambda(\RR^n)$ be the corresponding space of distributions.
A distribution $t\in\Dcal'(\RR^n\setminus\{0\})$ with scaling degree 
$\sd(t)$ has a unique extension $\bar{t}\in\Dcal_\lambda'(\RR^n)$, $\lambda=\sd(t)-n$, which satisfies the condition
$\sd(\bar{t})=\sd(t)$.
\end{thm}

We will call $\bar{t}$ the \textit{direct extension}. With the requirement $\sd(\dot t)=\sd(t)$, the extension is unique for $\sd(t)< n$ and given by the direct extension\footnote{Note that the axioms for the (regularized) time-ordered products used in \cite{DF04,BDF09} 
or in this paper (sect.~\ref{sec:dimregS}), imply the condition $\sd(\dot t)=\sd(t)$.}. For $\sd(t)\geq n$, the condition  does not fix the extension. To treat this case
we introduce a regularization. 

\begin{df}[Regularization]\label{df:regularisation} Let $t\in\Dcal'(\RR^n\setminus\{0\})$ be a distribution with degree of 
divergence $\lambda:=\sd(t)-n\geq 0$, and let $\bar{t}\in\Dcal_\lambda'(\RR^n)$ be the direct extension of $t$. A family of distributions $\{t^\zeta\}_{\zeta\in\Omega\setminus\{0\}}$, $t^\zeta\in\Dcal'(\RR^n)$, with $\Omega\subset\CC$ a neighborhood of the origin, is called a regularization of $t$, if
\be\label{eq:regularization}
\forall g\in\Dcal_\lambda(\RR^n):\quad\lim_{\zeta\rightarrow0}\langle t^\zeta,g\rangle=\langle \bar{t},g\rangle\,.
\ee
The regularization $\{t^\zeta\}$ is called analytic, if for all functions $f\in\Dcal(\RR^n)$ the map
\be
\Omega\setminus\{0\}\ni\zeta\mapsto \langle t^\zeta,f \rangle
\ee
is analytic with a pole of finite order at the origin. The regularization $\{t^\zeta\}$ is called finite, if 
the limit $\lim_{\zeta\rightarrow 0}\langle t^\zeta,f\rangle\in\CC$ exists $\forall f\in\Dcal(\RR^n)$. 
\end{df}
Note that for a finite regularization, the limit
$\lim_{\zeta\rightarrow0}t^\zeta$ is indeed a solution $\dot t$ of the
extension problem, that is $\lim_{\zeta\rightarrow0}\langle
t^\zeta,h\rangle=\langle t,h\rangle\quad\forall h\in\Dcal(\RR^n\setminus\{0\})$
and $\sd(\lim_{\zeta\rightarrow0}t^\zeta)=\sd(t)\ $.\footnote{To verify the latter
let $\dot t$ be an extension with $\sd(\dot t)=\sd(t)\ $. Writing $\Delta t:=\dot t-
\lim_{\zeta\rightarrow0}t^\zeta=\sum_\gamma C_\gamma\partial^\gamma\delta$, it follows from
$\Delta t\vert_{\Dcal_\lambda(\RR^n)}=0$ that $C_\gamma=0\,\,\,\forall|\gamma|>\lambda$ and,
hence, $\sd(\lim_{\zeta\rightarrow0}t^\zeta)=\sd(t)\ $.}

Any extension $\dot t\in\Dcal'(\RR^n)$ of $t$ with the same scaling degree is of 
the form $\langle\dot t,f\rangle=\langle \bar t,Wf\rangle$ with some projection,
\be\label{W-projection}
\begin{array}{rccl}
W:&\Dcal&\rightarrow & \Dcal_{\lambda} \\
  &  f  &\mapsto     &Wf:=f-\displaystyle{\sum_{|\gamma|\leq\sd(t)-n}w_\gamma\;\partial^\gamma f(0)}\ ,
\end{array}
\ee
given in terms of functions $w_\beta\in\Dcal(\RR^n)\ ,\,\,\,|\beta|\leq\sd(t)-n$ , fulfilling
\be\label{w-derivatives}
\partial^\gamma w_\beta(0)=\delta^\gamma_\beta\>\>\quad\quad\forall \gamma\in\NN_0^n
\ee
\cite[Lem.~B.2]{DF04}. Since $t^\zeta\in\Dcal'(\RR^n)$, we can write \eqref{eq:regularization} in the form
\be\label{regW-2}
\langle \bar{t},Wf\rangle=\lim_{\zeta\rightarrow0}\left[\langle t^\zeta,f\rangle - \sum_{|\gamma|\leq\sd(t)-n}\langle t^\zeta,w_\gamma\rangle\;\partial^\gamma f(0)\right].
\ee
In general,  the limits of the individual terms on the right hand side might not exist. However, if the regularization $\{t^\zeta,\zeta\in\Omega\setminus\{0\}\}$ is analytic, each term can be expanded in a Laurent series around $\zeta=0$, and since the overall limit is finite, the principal parts ($\pp$) of these Laurent series must coincide,
\be\label{pp=pp}
\forall f\in\Dcal(\RR^n):\quad
\pp(\langle t^\zeta,f\rangle) = \sum_{|\gamma|\leq\sd(t)-n}\pp(\langle t^\zeta,w_\gamma\rangle)\;\partial^\gamma f(0)\,.
\ee
Note that $\pp(\langle t^\zeta,w_\gamma\rangle)$ is independent of the choice of $w_\gamma$, because
 $\pp(t^\zeta)$ is a linear combination of derivatives of $\delta(x)$ and 
all information about $w_\gamma$ that is used is \eqref{w-derivatives}.  
We thus have proven
\begin{lemma}\label{lem:pp-local}
The principal part of any analytic regularization $\{t^\zeta\}$ of a distribution $t\in\Dcal'(\RR^n\setminus\{0\})$ is a local distribution of order $\sd(t)-n$, i.e.,
\be\label{pp=local}
\pp(t^\zeta) = \sum_{|\gamma|\leq\sd(t)-n}C_\gamma(\zeta)\;\delta^{(\gamma)}\,.
\ee
In the derivation above we have $C_\gamma(\zeta)=(-1)^{|\gamma|}\pp(\langle t^\zeta,w_\gamma\rangle)$.
\end{lemma}
Alternatively, the latter formula for $C_\gamma(\zeta)$ can be obtained directly from \eqref{pp=local}
by applying it to $w_\gamma$ and using \eqref{w-derivatives}.

\begin{cor}[Minimal Subtraction]\label{cor:MS-same-sd}
The regular part ($\rp=1-\pp$) of any analytic regularization $\{t^\zeta\}$ of a distribution $t\in\Dcal'(\RR^n\setminus\{0\})$ defines by
\be\label{def:MS}
\langle t^\MS,f\rangle :=\lim_{\zeta\rightarrow0} \rp(\langle t^\zeta,f\rangle)
\ee
an extension of $t$ with the same scaling degree, $\sd(t^\MS)=\sd(t)$.
The extension $t^\MS$ defined by (\ref{def:MS}) is called 'minimal subtraction'.
\end{cor}
In traditional terminology $(-1)\pp(t^\zeta)$ is a 'local counter term'.
\begin{proof}
It follows directly from (\ref{regW-2})-(\ref{pp=pp}) that any extension $\dot t$ of $t$ with the same scaling degree
can be written as
\begin{align*}
\langle \dot t,f\rangle=
\langle \bar{t},Wf\rangle&=\lim_{\zeta\rightarrow0}\left[\langle t^\zeta,f\rangle - \sum_{|\gamma|\leq\sd(t)-n} 
\left[\pp\langle t^\zeta,w_\gamma\rangle+\rp\langle t^\zeta,w_\gamma\rangle\right]\;\partial^\gamma f(0)\right].\\
&=\langle t^\MS,f\rangle
-\lim_{\zeta\rightarrow0}\sum_{|\gamma|\leq\sd(t)-n} \rp(\langle t^\zeta,w_\gamma\rangle)\;\partial^\gamma f(0)\,.
\end{align*}
Obviously $t^\MS$ differs from $\dot t$ by a local distribution of lower or equal scaling degree.
\end{proof}
\begin{cor}
For finite $\zeta$ the projection to the regular part of any analytic regularization $\{t^\zeta\}$ can be realized as a $W$-projection up to a term of order $\zeta$, i.e., there exists a projection $W^\MS:\Dcal(\RR^n)\rightarrow\Dcal_\lambda(\RR^n)$, $\lambda=\sd(t)-n$, such that
\be
\forall f\in\Dcal(\RR^n):\quad \rp\langle t^\zeta, f\rangle=\langle t^\zeta,W^\MS f\rangle+\Ocal(\zeta)\,.
\ee
\end{cor}
\begin{proof}
According to Corollary~\ref{cor:MS-same-sd} there is a projection $W^\MS:\Dcal(\RR^n)\rightarrow\Dcal_\lambda(\RR^n)$ such that
\be\label{eq:W-MS}
\langle t^\MS,f \rangle=\langle \bar t,W^\MS f\rangle
\ee
It follows for the regular part for all $f\in\Dcal(\RR^n)$,
\begin{align*}
\rp\langle t^\zeta, f\rangle
&=\bigg\langle t^\zeta-\pp(t^\zeta), W^\MS f+\sum_{|\gamma|\leq\sd(t)-n} w_\gamma^\MS f^{(\gamma)}(0) \bigg\rangle\\
&=\langle t^\zeta, W^\MS f\rangle + \sum_{|\gamma|\leq\sd(t)-n} \rp\langle t^\zeta,w_\gamma^\MS\rangle f^{(\gamma)}(0)\,,
\end{align*}
since $\pp\langle t^\zeta,W^\MS f\rangle=0$ by \eqref{pp=local}. The left hand side as well as the first term on the right 
hand side tend to $\langle t^\MS,f\rangle$ as $\zeta\rightarrow0$, cf.~(\ref{def:MS}) and \eqref{eq:regularization}, (\ref{eq:W-MS}). 
Hence the remaining sum on the right hand side needs to vanish in this limit, and since it is the regular part of a Laurent series it is at least of order $\zeta$.
\end{proof}

By means of the results above, the statements made at the beginning of section \ref{sec:forest-MS} 
can be illustrated on the level of the numerical distributions $t=t_{\Gamma}^{\beta}$ introduced 
in \eqref{time ordered functions}.
We first note that $t^\zeta\mapsto -\pp(t^\zeta)$ is 
just the action in terms of the numerical distributions of the projectors
$R_I$, $I=V(\Gamma)$ in the Forest Formula \eqref{EGforest}, 
because $\Scal^{\zeta\,(n)}$ depends on $\zeta$ only through 
$t^\zeta$ (see sect.~\ref{sec:dimregS}). The fact that $t^\MS$ is a 
renormalization of $t$ admitted by the Epstein-Glaser
axioms (Corollary \ref{cor:MS-same-sd}), reflects that $\Scal_{\MS}$ is a solution of these axioms.

Moreover, since in $(\Scal^{\zeta}\circ \Zcal^{\zeta}_{\MS})^{(n)}$ the action of $\Zcal^{\zeta}_{\MS}$
is the addition of the divergent counter terms for all contributing diagrams and for all their
subdiagrams, and since, in terms of the pertinent numerical distributions $t^\zeta$, 
these counter terms are given by $(-1)\pp(t^\zeta)$, and since $\pp(t^\zeta)$ is a 
{\it local} distribution with $\sd(\pp(t^\zeta))\leq\sd(t)$
(Lemma \ref{lem:pp-local}), we understand on the level of numerical distributions why
$\Zcal^{\zeta}_{\MS}\in\mathscr{R}$ (see \eqref{R-locality}). 
\subsection{Dimensionally Regularized Two Point Function}\label{sec:DimRegHadamard}
As in the unregularized case, the regularized Wightman two point function $\De_+^\zeta$
differs from $\De_F^\zeta$ \eqref{regFeynprop} only by a change of the boundary value prescription:
\be\label{regWightman}
\Delta^{\zeta}_+(x):=\mu^{2\zeta}\lim_{\epsilon\searrow 0}\wcal^{d-2\zeta}(x^2-i\epsilon x^0)\ .
\ee
As for $\De_F^\zeta$, we conclude that $\WF(\De_+^\zeta)\subset \WF(\Delta_+)$. Due to 
\be\label{eq:RegularizedCausalityDistributions}
\De_F^\zeta(x)=
\begin{cases}
\De_+^\zeta(x) \quad & \text{if}\quad x\notin \bar V_-\\
\De_+^\zeta(-x) \quad & \text{if}\quad x\notin \bar V_+
\end{cases}\ ,
\ee
causality of the regularized time-ordered products $\TT^{\bzeta}_n$ can be postulated in the 
Bogoliubov-Shirkov way \cite{BS59} and with that the family $(\TT^{\bzeta}_n)$ can be 
constructed inductively by a version of the Epstein-Glaser method; this is done in the next subsection.

Under a simultaneous scaling of $x$ and $m$, $\De_+^\zeta$ and $\De_F^\zeta$ are homogeneous,
\be\label{hom:scaling}
\rho^{d-2-2\zeta}\,\De_{j,\rho^{-1}m}^\zeta(\rho x)=\De_{j,m}^\zeta(x)\ ,\quad j=+,F\ ,
\ee
since $\wcal^{d-2\zeta}(z^2)$ has this property. 

We now use the relation between Bessel functions
\be
K_\nu=\frac{\pi}{2\sin(\nu\pi)}[I_{-\nu}-I_\nu]\quad(\nu\in\CC\setminus\ZZ)\quad
\ee 
and the fact that $I_{\nu}$ is of the form
\be
I_{\nu}(z)=z^{\nu}F_{\nu}(z^2)
\ee
with the entire function
\be
F_{\nu}(z^2)=2^{-\nu}\sum_{k=0}^\infty\frac1{k!\Gamma(\nu+k+1)}\left(\frac{z^2}{4}\right)^{k} \ .
\ee
Inserting these relations into the formula for the 2-point function we obtain the decomposition
\be\label{D=H+C}
\De_{j,m}^\zeta=H_j^{m,\zeta}+C_m^\zeta
\ee
where
\be\label{Creg}
C_m^\zeta(x)=-c(d-2\zeta)m^{d-2}\left(\frac{\mu}{m}\right)^{2\zeta}F_{\frac{d}{2}-1-\zeta}(-m^2x^2)
\ee 
\be\label{Hreg}
H_j^{m,\zeta}(x)=c(d-2\zeta)\mu^{2\zeta}\left((-x^2)^{1-\frac{d}{2}+\zeta}\right)_jF_{\zeta-\frac{d}{2}+1}(-m^2x^2)
\ee
with $c(d-2\zeta):=(2\pi)^{\zeta-\frac{d}{2}}\frac{\pi}{2\sin((\frac{d}{2}-1-\zeta)\pi)}$. The index $j=+,F$ denotes as above the appropriate boundary values. Note that 
the zeroes of the sine function at multiples of $\pi$ produce poles at $\zeta=\frac{d}{2}+n, n\in\ZZ$ in the above decomposition which cancel in the sum \eqref{D=H+C}. 
 
We observe that $H_j^{m,\zeta}$ is a smooth function of the mass $m$ and $C^{m,\zeta}$ is a smooth function of the position $x$. Both terms satisfy the homogeneous scaling \eqref{hom:scaling}. $H_F^{m,\zeta}$ is the Feynman type propagator corresponding to the Hadamard function which was already used in \cite{BDF09} (see also \cite{Kel10a}).

The interpretation of  $\De_+^\zeta$ as the dimensionally regularized 2-point function (in spite of the fact that it is a distribution in $d$ dimensions) may be justified by the fact that it solves an appropriately deformed
version of the Klein-Gordon equation. This may be useful for the discussion of symmetries (as current conservation or gauge invariance cf.~\cite{BD08,FRb}) for the
dimensionally regularized amplitudes.

\begin{lemma}\label{mod-KG}
Let $d_\zeta:=d-2\zeta\ $, $t:=x^0\ $, $r:=\sqrt{\sum_{i=1}^{d-1} (x^i)^2}$
and let the $\zeta$-dependent functions $f_\zeta$ and $G_\zeta$ be related by
\be
f_\zeta(z)=z^{-d_\zeta/2+1}\,G_\zeta(mz)\ .
\ee
For $r\not=0$ we introduce the 'wave operator
in $d_\zeta$-dimensions':  
\be\label{reg-wellenop}
 \square_d^\zeta:=\partial_t^2-\partial_r^2
-\frac{d_\zeta-2}{r}\partial_r-\frac{1}{r^2}\Delta_{S^{d-2}}\ .
\ee
with the Laplacian $\Delta_{S^{d-2}}$ on the $(d-2)$-sphere.

\begin{itemize}
\item[(a)] For $x^2=t^2-r^2<0$ it holds: $F_\zeta(x):=f_\zeta(\sqrt{-x^2})$ solves
the 'Klein-Gordon equation in $d_\zeta$-dimensions', i.e. 
\be\label{eq:mod-KG}
(\square_d^\zeta+m^2)\, F_\zeta(x)=0\ ,
\ee
if and only if $G_\zeta(u)$ is a solution of the modified Bessel equation of order 
$d_\zeta/2-1$,
\be\label{mod-Bessel}
G''_\zeta(u)+\frac{G'_\zeta(u)}{u}+G_\zeta(u)\Bigl(1+\frac{(d_\zeta/2-1)^2}{u^2}\Bigr)=0\ .
\ee
\item[(b)] $\De_+^\zeta(x)$ solves the 'Klein-Gordon equation in $d_\zeta$-dimensions'
\eqref{eq:mod-KG} for all $x$ with $r\not=0$.
\end{itemize}
\end{lemma}


\begin{proof} {\it (a) and (b):} The statement (a) is obtained straightforwardly 
by inserting the definitions and computing the derivatives. Since 
\be
\De^\zeta_+(x)=\lim_{\epsilon\downarrow 0}f_\zeta(\sqrt{r^2-t^2+it\epsilon})
\ee
with a pertinent function $G_\zeta$ solving \eqref{mod-Bessel}, part (a) immediately yields
$(\square_d^\zeta+m^2)\,\De^\zeta_+(x)=0$ for $x^2<0$. For $x^2\geq 0$ the calculation in the proof of (a)
has to be supplemented by the $i\epsilon$-terms, the final limit $\epsilon\downarrow 0$ is harmless.
\end{proof}

\subsection{Dimensionally Regularized Time-ordered products}\label{sec:dimregS}
In contrast to the situation described in (\cite[Sect.~5.2]{BDF09}) and (\cite[Sect.~4]{DWilson}),
the regularized Feynman propagator $\De^\zeta_F\in\Dcal^\prime(\MM)$ is {\it not} a smooth function.
Actually, we are not aware of any analytic regularization which yields smooth propagators.
Hence, the construction of the regularized time-ordered products
involves non-direct extensions of distributions.

The aim of this section is to construct a unique family of linear maps 
\be\label{Sreg}
\Tcal_n^\bzeta:\Floc^{\otimes n}\rightarrow\Fcal
\ee
perturbatively by Epstein-Glaser induction (to simplify the notations we write $\Floc$ and $\Fcal$ for $\Floc[[\hbar]]$ and $\Fcal[[\hbar]]$ resp.). The construction has to be done in such a way
that each $\TT_n^{\,\bzeta}$ in the perturbative expansion is a meromorphic function of 
$N:= {n\choose 2}$ complex parameters $\zeta_{ij}$, i.e. each order in $\hbar$ is meromorphic.
We choose different parameters $\zeta_{ij}$
for each bidifferential operator in the formal expression
\be\label{Sreg-unren}
T^{\bzeta,\,\mathrm{unren}}_{n}:=\exp\sum_{1\leq i<j\leq n} D_{ij}^{\zeta_{ij}}\ ,\quad
D_{ij}^{\zeta_{ij}}:=\langle \hbar \De_F^{\zeta_{ij}},\tfrac{\delta^2}{\delta\varphi_i\delta\varphi_j}\rangle,\,
\ee
which we now want to make precise with the use of homogeneous extensions of distributions ('$\mathrm{unren}$' stands for 'unrenormalized'). We can expand the exponential in \eqref{Sreg-unren} in terms of graphs by means of \eqref{time:ord}, so  we can construct $T^{\bzeta}_n$ as a sum of  $T_\Gamma^{\bzeta}$, with $\Gamma\in\Gcal_n$ (the set of all graphs with vertices $\{1,\dots n\}$). Each expression $T^{\bzeta}_\Gamma$ can be obtained by recursively extending $t_{\Gamma}$ given by the formula \eqref{SGamma} to a distribution defined on the space $\mathcal{D}(\Delta_{\Gamma},Y_{\Gamma})$ introduced in Section \ref{EG} . The family $\bzeta$ contains one regularization parameter for each pair of vertices. 
We write $\bzeta:=(\zeta_{ij})_{1\leq i<j\leq n}\in\CC^N$. 
The regularized time-ordered products $\Tcal_n^\bzeta$ are given by $m_n\circ T^{\bzeta}_{n}$. We will show in this section that the latter can be constructed in such a way that certain properties, similar to Epstein-Glaser axioms are satisfied. These properties can be specified equivalently on the level of $T^{\bzeta}_{n}$ or $T^{\bzeta}_{\Gamma}$ and we will make use of both possibilities, dependent on notational convenience.
The axioms which we assume are the following (compare with \cite{DF04} and \cite{BDF09}):
\begin{itemize}
\item{\bf Starting element}. $T_1^{\bzeta}=\id\ $,
\item{\bf Causality}. Let $F_1\,\dots,F_n$ be local functionals such that $F_1,\dots,F_k$ have supports later than the supports of $F_{k+1},\dots,F_n$. Let us denote by $I$ the index set $\{1,\dots,k\}$ and by $\bzeta_I$ the family of parameters 
$\zeta_{ij}$, where $i,j\in I$. Similarly, elements of $\bzeta_{I^c}$ will have $i,j\in I^c$ and elements of $\bzeta_{II^c}$ satisfy: $i\in I$, $j\in I^c$. Together they form the set of parameters $\bzeta=(\bzeta_I,\bzeta_{I^c},\bzeta_{II^c})$.
 The condition of causality is the requirement that
\be\label{caus}
T_n^{\bzeta}(F_1,\dots,F_n)=\exp(\sum\limits_{i\leq k\atop j> k}D_{ij}^{\bzeta_{II^c}})\ T_k^{\bzeta_I}(F_1,\dots,F_k)\otimes T_{n-k}^{\bzeta_{I^c}}(F_{k+1},\dots,F_n)\,.
\ee

\item{\bf $\ph$-Locality}. We require that $T_n$ is, in $L$-th order of $\hbar$, a functional differential operator on $\Ecal(\M)$ of order $2L$ (see \cite{BDF09} for details).

\item{\bf Field Independence}. For every $k=1,\ldots,n$ we require that
  $\langle\frac{\delta}{\delta \ph_k} T_n^\bzeta(F_1,\dots,F_n),\psi\rangle=T_n^\bzeta(F_1,\dots,\langle\frac{\delta F_k}{\delta \ph_k},\psi\rangle,\dots,F_n),\ $ with $F_1,\dots,F_n\in\Floc,\,\psi\in\Ecal(\MM)\ $.

As explained in \cite{BDF09}, $\ph$-Locality and  Field Independence imply that 
$T^{\,\bzeta}_n(F^{\otimes n})$ can be expanded in the fields as follows:
\be\label{causWick}
T^{\,\bzeta}_n(F^{\otimes n})(\ph_1,\ldots,\ph_n)=
\sum_{\alpha,\beta}\langle t_{\alpha}^{\bzeta,\beta},f^{\alpha_1}_{\beta_1}(\ph_1)\otimes\cdots f^{\alpha_n}_{\beta_n}(\ph_n)\rangle  \,
\ee
where the test functions $f^{\alpha_i}_{\beta_i}(\ph)$ are defined in \eqref{dF} and the numerical distribution 
$t^{\bzeta,\beta}_{\alpha}=\sum_{\Gamma}t^{\bzeta,\beta}_{\Gamma}$ is a time-ordered product of balanced fields\footnote{Balanced fields are local field polynomials 
$A(x) = P(\partial_1,\dots,\partial_n)\varphi(x_1)\cdots\varphi(x_n)\big|_{x_1=\dots=x_n=x}$.
Here $P$ is a polynomial, with the peculiarity that $P$ 
depends only on the differences of variables $(\partial_i - \partial_j)$ (``relative derivatives'').
Balanced fields, originally introduced in \cite{BOR02}, were used in \cite{DF04} 
in order to fulfill Stora's Action Ward Identity (AWI). The latter guarantees that the time-ordered product 
depends only on local interaction {\it functionals} $F$, and not on 
the choice of a corresponding Lagrangian.} at $\ph=0$,
 where the sum runs over all graphs with vertex set $V(\Gamma)=\{1,\dots,n\}$
and $\alpha_i$ lines at the vertex $i$, $i=1,\dots,n$.

This formula is a generalization of the causal Wick expansion given in \cite{EG73}. 
We point out that the r.h.s.~of \eqref{causWick} depends 
on $\bzeta$ only through  the numerical distributions $t^{\bzeta,\beta}_{\alpha}$.

\item{\bf Translation Invariance}. This axiom can be expressed by the requirement that the numerical 
distributions  $t^{\bzeta,\beta}_{\alpha}$ appearing in the field expansion
\eqref{causWick} depend only on the relative coordinates $(x_1-x_n,...,x_{n-1}-x_n)$.

\item{\bf Smoothness in $m^2$}. To formulate the requirement of the smoothness in mass we make use of the decomposition of the Feynman propagator $\De_{F}^{m,\zeta}$ into $H_{F}^{m,\zeta}$ and $C^{m,\zeta}$. Let $D_{ij}^{C}:=\langle \hbar C^{m,\zeta_{ij}},\tfrac{\delta^2}{\delta\varphi_i\delta\varphi_j}\rangle$. We can ``factor out'' the powers of $C^{m,\zeta}$ from the regularized time-ordered products by applying the following equivalence relation:
\be\label{TH:TF}
T^{\bzeta}_{H,n}(F_1,\dots, F_n):=\exp\Big(-\sum_{i<j}D_{ij}^{C}\Big)\circ T^{\bzeta}_n (F_1,\dots, F_n)\,.
\ee
Let us now explain what this operation means in terms of Feynman graphs. First we decompose a given graph into a sum of graphs
that have $H_{F}^{m,\zeta}$  or $C^{m,\zeta}$ assigned to lines. For example, the setting sun graph can be written as
\[
\begin{tikzpicture}[thick,scale=1.5]
\useasboundingbox (0,0) rectangle (1,0.6);
\filldraw (0,0) circle (1pt);
\filldraw (0.8,0) circle (1pt);
\draw (0,0) edge [out=80,in=100] node[above] {} (0.8,0);
\draw (0,0) -- node[above] {} (0.8,0);
\draw (0,0) edge [out=-80,in=-100] node[above] {} (0.8,0);
\end{tikzpicture} 
\,=
\begin{tikzpicture}[thick,scale=1.5]
\useasboundingbox (0,0) rectangle (1,0.6);
\filldraw (0,0) circle (1pt);
\filldraw (0.8,0) circle (1pt);
\draw (0,0) edge [out=80,in=100] node[above=-1.5pt] {\footnotesize$C$} (0.8,0);
\draw (0,0) -- node[above=-1.5pt] {\footnotesize$C$} (0.8,0);
\draw (0,0) edge [out=-80,in=-100] node[above=-1.5pt] {\footnotesize$C$} (0.8,0);
\end{tikzpicture}
\ 
+
3\ \begin{tikzpicture}[thick,scale=1.5]
\useasboundingbox (0,0) rectangle (1,0.6);
\filldraw (0,0) circle (1pt);
\filldraw (0.8,0) circle (1pt);
\draw (0,0) edge [out=80,in=100] node[above=-1.5pt] {\footnotesize$C$} (0.8,0);
\draw (0,0) -- node[above=-1.5pt] {\footnotesize$H$} (0.8,0);
\draw (0,0) edge [out=-80,in=-100] node[above=-1.5pt] {\footnotesize$C$} (0.8,0);
\end{tikzpicture}
\ 
+
3\  \begin{tikzpicture}[thick,scale=1.5]
\useasboundingbox (0,0) rectangle (1,0.6);
\filldraw (0,0) circle (1pt);
\filldraw (0.8,0) circle (1pt);
\draw (0,0) edge [out=80,in=100] node[above=-1.5pt] {\footnotesize$H$} (0.8,0);
\draw (0,0) -- node[above=-1.5pt] {\footnotesize$H$} (0.8,0);
\draw (0,0) edge [out=-80,in=-100] node[above=-1.5pt] {\footnotesize$C$} (0.8,0);
\end{tikzpicture}
\ +
\  \begin{tikzpicture}[thick,scale=1.5]
\useasboundingbox (0,0) rectangle (1,0.6);
\filldraw (0,0) circle (1pt);
\filldraw (0.8,0) circle (1pt);
\draw (0,0) edge [out=80,in=100] node[above=-1.5pt] {\footnotesize$H$} (0.8,0);
\draw (0,0) -- node[above=-1.5pt] {\footnotesize$H$} (0.8,0);
\draw (0,0) edge [out=-80,in=-100] node[above=-1.5pt] {\footnotesize$H$} (0.8,0);
\end{tikzpicture}
\]
Next, we write each such graph as a product of two graphs with only one kind of lines, for example:
\[
\begin{tikzpicture}[thick,scale=1.5]
\useasboundingbox (0,-0.1) rectangle (1,0.6);
\filldraw (0,0) circle (1pt);
\filldraw (0.8,0) circle (1pt);
\draw (0,0) edge [out=80,in=100] node[above=-1.5pt] {\footnotesize$C$} (0.8,0);
\draw (0,0) -- node[above=-1.5pt] {\footnotesize$H$} (0.8,0);
\draw (0,0) edge [out=-80,in=-100] node[above=-1.5pt] {\footnotesize$C$} (0.8,0);
\end{tikzpicture}\ =\ \begin{tikzpicture}[thick,scale=1.5]
\useasboundingbox (0,-0.1) rectangle (1,0.6);
\filldraw (0,0) circle (1pt);
\filldraw (0.8,0) circle (1pt);
\draw (0,0) edge [out=80,in=100] node[above=-1.5pt] {\footnotesize$C$} (0.8,0);
\draw (0,0) edge [out=-80,in=-100] node[above=-1.5pt] {\footnotesize$C$} (0.8,0);
\end{tikzpicture}\cdot\ \begin{tikzpicture}[thick,scale=1.5]
\useasboundingbox (0,-0.1) rectangle (1,0.6);
\filldraw (0,0) circle (1pt);
\filldraw (.8,0) circle (1pt);
\draw (0,0) -- node[above=-1.5pt] {\footnotesize$H$} (0.8,0);
\end{tikzpicture}
\]
Similarly to \cite{BDF09}, we require the maps $T^{\bzeta}_{H,n}$ to be smooth in $m^2\in\RR$. 
Since one can switch between 
$T^{\bzeta}_{H,n}$ and $T^{\bzeta}_n$ using the map $\exp\,{\sum_{i<j}D_{ij}^{C}}$, this requirement is a condition that affects also $T^{\bzeta}_n$. The contribution to $T^{\bzeta}_{H,n}$ coming from a graph $\Gamma$ will be denoted by  $T^{\bzeta}_{H,\Gamma}$.

\item{\bf Scaling}. Both the regularized Feynman propagator $\De^{m,\zeta}_F$ and $H_{F}^{m,\zeta}$ satisfy the scaling property \eqref{hom:scaling}, so it is natural to require a corresponding scaling behavior from $T^{\bzeta}_{H,n}$ and $T^{\bzeta}_n$. Following \cite{DF04,BDF09}, we define a map $\sigma_\rho:\Fcal\rightarrow \Fcal$, which acts as the scaling transformation:
\[
\sigma_\rho (F)(\ph):=F(\rho^{\frac{2-d}{2}}\ph_\rho)\,,\qquad \ph_\rho(x)=\ph(\rho^{-1}x)\,.
\] 
For $F_1\dots F_n\in \Fcal$ with disjoint supports
\be\label{nonren:scaling}
\sigma_\rho\circ T^{m,\bzeta}_n\circ \sigma_\rho^{-1}(F_1,\dots,F_n)=\exp\Big(\sum_{1\leq i< j\leq n}\rho^{2\zeta_{ij}}D^{\rho m}_{ij}\Big)(F_1,\dots,F_n)
\ee
holds, where we exhibited the dependence on the mass $m$. To formulate the scaling condition, it is convenient to work on the level of graphs. In the expression \eqref{nonren:scaling} we get a factor $\rho^{2\zeta_{ij}}$ for each line joining vertices $i$ and $j$. We want the extended $ T^{m,\bzeta}_n$ and $ T^{m,\bzeta}_{H,n}$ to behave in the same way, so we require that
\be\label{scalingH}
\sigma_\rho\circ T^{m,\bzeta}_{H,\Gamma}\circ \sigma_\rho^{-1}=\rho^{2 \mathbf{l}\bzeta}\,T^{\rho m,\bzeta}_{H,\Gamma}\,,
\ee
and the same for $T^{m,\bzeta}_{\Gamma}$. In the formula above, $l_{ij}$ is the number of lines connecting vertices $i$ and $j$ and $\mathbf{l}\bzeta$ is the scalar product $\mathbf{l}\bzeta:=\sum_{i,j\in V(\Gamma)}l_{ij}\zeta_{ij}$. 
The formula above may be illustrated by the following example: 
{\small \begin{example}
\begin{align}
&\sigma_\rho\circ
D_{12}^{m,\zeta_{12}}\,D_{23}^{m,\zeta_{23}}\circ\sigma_\rho^{-1}
(\ph_1^2(x), \ph_2^2(y),\ph_3^2(z))\notag\\
&\quad\quad= 8\hbar^2\,\rho^{2(d-2)}\,\Delta_{F}^{m,\zeta_{12}}
(\rho(x-y))\,\Delta_{F}^{m,\zeta_{23}}(\rho(y-z))\,
\sigma_\rho(\rho^{d-2}\ph_1(\rho x)\ph_3(\rho z))\notag\\
&\quad\quad= \rho^{2(\zeta_{12}+\zeta_{23})}\, D_{12}^{\rho m,\zeta_{12}}\,
D_{23}^{\rho m,\zeta_{23}}(\ph_1^2(x), \ph_2^2(y),\ph_3^2(z)) \end{align}
\end{example}}

\end{itemize} 
We will now show that the given axioms determine the family 
$(T_n^\bzeta)$ {\it uniquely}, for an appropriate choice of the parameters $\bzeta$. First we construct the family $(T_{H,n}^\bzeta)$ by the Epstein-Glaser induction. 
 Using the causal factorization and the field expansion, in each order $n$, we reduce the problem to the extension of a numerical distribution defined everywhere outside of the thin diagonal. The crucial property that allows us to do this is the fact that the propagators $H^{m,\bzeta}_F$ are symmetric for spacelike points and therefore the definition of $(T_{H,n}^\bzeta)$ doesn't depend on the way in which we split $F_1,\dots,F_n$ into an earlier and later supported set on the r.h.s. of \eqref{caus}. The scaling behavior of the numerical distributions is obtained from the formula \eqref{scalingH}, after inserting the field expansions \eqref{causWick} of functionals. See \cite{DF04} for details of this construction. For a given graph $\Gamma$, the scaling degree of a numerical coefficient $t_{H,\Gamma}^{\bzeta,\beta}$ is given by
\be\label{kappa:zeta}
\kappa^{\bzeta,\beta}=\sum\limits_{i<j} l_{ij}(d-2-2\zeta_{ij})+|\beta|\,.
\ee
We choose the parameters $\zeta_{ij}$ in such a way, that  $\kappa^{\bzeta,\beta}\notin\NN_0+d(|V(\Gamma)|-1)$ and $\Re(\kappa^{\bzeta,\beta})<\kappa^{\mathbf{0},\beta}+1$, where $\kappa^{\mathbf{0},\beta}:=\sum_{i<j} l_{ij}(d-2)+|\beta|$.  The reason for the former condition will be seen in \eqref{diffrenhom0}. The latter condition guarantees that the regularization doesn't make ${t_{H,\Gamma}^{\bzeta,\beta}}$ too singular.
 If these conditions are fulfilled, we obtain a unique homogeneous extension with the same degree $\kappa^{\bzeta,\beta}$, which is smooth in $m^2$. The uniqueness follows immediately from the fact that two homogeneous extensions would differ by a sum of derivatives of the $\delta$-function multiplied by non-integer powers of $m$, thus violating the smoothness condition. The explicit construction of such an extension will be given in section \ref{extension-hom}. In section \ref{examples} we illustrate the inductive procedure presented here on the level of single diagrams.

Having constructed the family $(T_{H,n}^\bzeta)$ as a unique solution to our extension problem, we can obtain $(T_n^\bzeta)$ by applying \eqref{TH:TF}. 
The maps $T_n^\bzeta$ constructed  here have some additional useful
{\bf properties}:\footnote{In the exact theory (i.e.~for $\bzeta ={\bf 0}$), 
these properties play the role of renormalization 
conditions, they are part of the axioms \cite{EG73,DF04,BDF09}.}
\begin{itemize} 
\item{\bf  Lorentz Covariance}. Since $\De^\zeta_+$ and $\De^\zeta_F$ are of the form
$\De^\zeta_+(x)=\wcal^\zeta(x^2-i\epsilon x^0)$ and $\De^\zeta_F(x)=\wcal^\zeta(x^2-i\epsilon)$ resp.,
they are  Lorentz invariant. This is the origin of Lorentz Covariance of $T_n^\bzeta$:
\be
\beta_L(T_n^\bzeta(F_1,\dots,F_n))=T_n^\bzeta(\beta_L(F_1),\dots,\beta_L(F_n))\quad\quad
    \forall L\in {\cal L}_+^\uparrow\ ,
\ee
where $\beta$ is the  natural automorphic action of the Lorentz group ${\cal L}_+^\uparrow$
 on $\Fcal$.
\item{\bf Unitarity}.
In the exact theory ($\bzeta ={\bf 0}$) one wants the relation $\bar S(-F)=(S(F))^{-1}$
to hold true, where
$\bar S(F):=\overline{S(\bar F)}$. In our formalism for the regularized 
time-ordered products, the corresponding property can be formulated as
\be\label{eq:unitarity}
(-1)^n\,\bar{T}_n^\bzeta=\hat{T}_n^\bzeta\ ,
\ee
where $\bar{T}_n^\bzeta(F^{\otimes n}):=\overline{T_n^{\bar\bzeta}(\bar F^{\otimes n})}$ and
\be\label{hatT}
\hat{T}_n^\bzeta:=
\sum_{P=(I_1,...,I_r)\in\mathrm{Part}(\{1,\dots,n\})}(-1)^r \exp\Bigl(
\sum_{i\in I_k,\,j\in I_l\,\mathrm{with}\,k<l}D_{ij}^+\Bigr)\,T^{\bzeta_{I_1}}_{|I_1|}\otimes\cdots\otimes
T^{\bzeta_{I_r}}_{|I_r|}
\ee
with $D_{ij}^+:=\langle \hbar \De_+^{\zeta_{ij}},\tfrac{\delta^2}{\delta\varphi_i\delta\varphi_j}\rangle$.
Similarly to the exact theory (cf. \cite{EG73,Sch89})\footnote{The fact that we work with different
$\zeta$'s does not complicate the calculations, because the propagators depend on $(i,j)$ already 
via their argument $(x_i-x_j)$.}, 
$\hat{T}_n^\bzeta$ satisfies anti-causal factorization, i.e.~the
$T^\bzeta$-factors on the r.h.s.~of \eqref{caus} appear in reversed order.

This property holds also for $\bar{T}^\bzeta_n$, because the underlying propagator is the
anti-Feynmann propagator, that is
\be\label{barT}
\bar{T}^{\bzeta,\mathrm{unren}}_n:=\exp\sum_{1\leq i<j\leq n} D^{AF}_{ij}\ ,\quad
D^{AF}_{ij}:=\langle \hbar \De_{AF}^{\zeta_{ij}},\tfrac{\delta^2}{\delta\varphi_i\delta\varphi_j}\rangle,\,,
\ee
where
\begin{align}
\De_{AF}^\zeta(x)&:=\Theta(x^0)\De_+^\zeta(-x)+\Theta(-x^0)\De_+^\zeta(x)\notag\\
&=\wcal^\zeta(x^2+i\epsilon)=\overline{\wcal^{\bar\zeta}(x^2-i\epsilon)}=
\overline{\De_{F}^{\bar\zeta}(x)}\label{AFprop}\,.
\end{align} 

From the anti-causal factorization of both, $\bar{T}_n^\bzeta$ and $\hat{T}_n^\bzeta$, we
conclude that in the inductive Epstein-Glaser construction unitarity \eqref{eq:unitarity}
can possibly be violated only in the extension to the thin diagonal. However, since both sides of 
\eqref{eq:unitarity} are {\it uniquely} extended by homogeneity, also the extensions must agree.

\item{\bf Field Equation}. Let $G=\int dx\,\ph(x)\,h(x)$ (where $h\in\Dcal(\MM)$). 
By the field equation we mean the relation
\begin{align}\label{FE}
&T^{\,\bzeta}_n(G,F_1,\dots,F_{n-1})=G\otimes T^{\,\bzeta}_{n-1}(F_1,\dots,F_{n-1})\notag\\
&\quad\quad\quad+\sum_{i=1}^{n-1}\int dx\,h(x)\int dy\, \De_F^{\zeta_{0i}}(x-y)\,T^{\,\bzeta}_{n-1}
\Bigl(F_1,\dots,\tfrac{\delta F_i}{\delta\ph(y)},\dots,F_{n-1}\Bigr)\ .
\end{align}
The validity of the Field Equation is most easily shown by using the uniqueness of $T^{\bzeta}_n$. 
The right hand side of \eqref{FE} gives an alternative inductive definition of $T^{\bzeta}_n$ on the restricted domain
$\{G\otimes F_1\otimes...F_{n-1}\,|\,G=\int\ph\,h\ ,\,\,F_i\in \Floc\}$, which fulfills all the axioms. 
Therefore, the alternative definition \eqref{FE} of 
$T^{\bzeta}_n(G\otimes F_1\otimes...F_{n-1})$ must agree with the original one.
\item{\bf Meromorphicity}. The maps $T^{\,\bzeta}_n$ are meromorphic in $\bzeta$.
\end{itemize}
\subsection{Extension of homogeneously scaling distributions 
with non-integer degree}\label{extension-hom}

In this section we derive a general formula for differential renormalization of homogeneous 
distributions with a non-integer degree. In order to include the case of nonzero masses, we consider smooth distribution valued functions of $m^2$,
\[t:\RR\to \Dcal^\prime(\RR^l\setminus\{0\})\]
(i.e. for every test function $f\in\Dcal(\RR^l\setminus\{0\})$, the function $m^2\mapsto\langle t(m^2),f\rangle$ is smooth).
We assume that $t$ is homogeneous under simultaneous scaling of $m^2$ and the underlying coordinates $x_r,r=1,\dots,l$, 
i.e. 
\be\label{scaling}
(\sum_{r=1}^l Q_{r}\,\partial_r -m\partial_m)t=-\kappa\,t \,
\ee
with  $\kappa\not\in l+\NN_0$,
where $Q_r$ is the operator of multiplication with the function $x\mapsto x_r$.
Due to smoothness in $m^2$, the scaling degree of $t(m^2)$ is equal to $\mathrm{Re}\,\kappa$.
 
It is convenient to introduce a uniform notation and write $Q_0$ for the operator $\partial_m$.  Furthermore, we denote $P_r:= \partial_r$, $r=1,\dots,l$ and  $P_0$ is the multiplication by $-m$. Using this notation we define a sequence of operators $E_n$ by
\begin{align*}
E_0&:=1\\
E_{n+1}&:=\sum_{r=0}^l P_r E_n Q_r\,.
\end{align*}
We can think of $E_n$ as generalized Euler operators, hence the notation.
\begin{lemma}\label{lemma:scaling}
The scaling relation
\eqref{scaling} implies the formula
\be\label{diffrenhom}
t=\frac{1}{\prod_{j=0}^{n-1}(l+j-\kappa)} E_n t\,,\ \forall n\in\NN\,.
\ee
\end{lemma}
\begin{proof}We prove this by induction on $n$. The case $n=1$ is the scaling relation
\eqref{scaling} written in the form
\be\label{scaling1}
\sum_{r=0}^lP_r\Bigl(Q_r\, t\Bigr)=(l-\kappa)\,t\ .
\ee
Assuming \eqref{diffrenhom} to hold true for $n\leq k$, we take into account that
$Q_r t$ is homogeneous with degree $(\kappa-1)$
(i.e.~it satisfies \eqref{scaling1} correspondingly modified) and obtain
\begin{multline}
E_{k+1}t=\sum_{r_1...r_{k+1}}P_{r_{k+1}}\dots P_{r_{1}}\Bigl(Q_{r_1}\dots Q_{r_{k+1}}\,t\Bigr)=\\
=\sum_{r_1...}P_{r_{k+1}}\dots P_{r_2}\Bigl(\sum_rP_{r}\circ Q_r\circ
\Bigl(Q_{r_2}\dots Q_{r_{k+1}}\,t\Bigr)\Bigr)\notag=\\
=(l-(\kappa-k))\,\sum_{r_1...}P_{r_{k+1}}\dots P_{r_2}
\Bigl(Q_{r_2}\dots Q_{r_{k+1}}\,t\Bigr)=\notag\\
=(l+k-\kappa)\,\Bigl(\prod_{j=0}^{k-1}(l+j-\kappa)\Bigr)\,t\ ,
\end{multline}
which is \eqref{diffrenhom} for $n=k+1$.
\end{proof}
We will now use this lemma for defining extensions of distributions. Obviously, multiplication by $x_r$ reduces the scaling degree by 1. Since $t$ is a smooth function of mass, fulfilling \eqref{scaling}, the scaling degree is also reduced by 1 if we apply $\partial_m$. 
Let $\omega\in\ZZ$ be the minimal integer fulfilling
\be\label{omega}
\omega>\Re(\kappa)-l-1\ .
\ee
Now, choosing $n=\omega+1$ in \eqref{diffrenhom} we have
\be
\mathrm{sd}\bigl(Q_{r_1}\dots Q_{r_{\omega+1}}\,t\bigr)=
\Re(\kappa)-(\omega +1)< l\ .
\ee
Hence, $Q_{r_1}\dots Q_{r_{\omega+1}}\,t$ can be uniquely extended  (by the direct extension, 
see Thm.~\ref{thm:Extension-sd}) to a homogeneous distribution
\[
\overline{Q_{r_1}\dots Q_{r_{\omega+1}}\,t}\in\Dcal'(\RR^l)\,.
\]
Using differential renormalization, the unique homogeneous extension $\dot t\in\Dcal'(\RR^l)$ of $t$ is given by 
\be\label{diffrenhom0}
\dot t=\frac{1}{\prod_{j=0}^\omega (l+j-\kappa)}\,
\sum_{r_1...r_{\omega+1}}P_{r_{\omega+1}}\dots P_{r_1}
\Bigl(\overline{Q_{r_1}\dots Q_{r_{\omega+1}}\,t}\Bigr)\ .
\ee
It is now clear, why the assumption $\kappa\not\in l+\NN_0$ is needed. The massless case is easily obtained by setting $Q_0$ and $P_0$ to 0.
\begin{rem}[Almost homogeneous scaling distributions]\footnote{This remark 
is not relevant for our construction, but it may be useful in other instances.}
For {\it almost homogeneous scaling distributions $t\in\Dcal^\prime(\RR^l\setminus\{0\})$ 
with $\omega=0$}, the scaling relation
(which is now \eqref{almosthomscal}) can also be used to derive a formula for differential renormalization.
More precisely we assume that $t$ fulfills \eqref{almosthomscal} with degree $\kappa=l+z$,
where $0\not= z\in\CC$ and $\Re(z)<1$. Now we write \eqref{almosthomscal} in the form
\begin{align}
0&=(P\cdot Q+z)^{N+1}\,t\notag\\
&=z^{N+1}\,t+\sum_{k=1}^{N+1}\,{N+1\choose k} \,z^{N+1-k}
\sum_sP_s\Bigl(Q_s\circ(P\cdot Q)^{k-1}\,t\Bigr)\ ,
\end{align}
where $P\cdot Q:=\sum_{r=1}^lP_r\circ Q_r$. Since
$\mathrm{sd}\bigl(Q_s\circ(P \cdot Q)^{k-1}\,t\bigr)=\Re(\kappa)-1<l$
the unique, almost homogeneous extension can be written as
\be\label{almosthomscal1}
\dot t=-\sum_{k=1}^{N+1}\,{N+1\choose k}\,\frac{1}{z^k}\,
\sum_{s=1}^lP_s\Bigl(\overline{Q_s\,(P \cdot Q)^{k-1}\,t}\Bigr)
\in \Dcal^\prime(\RR^l)\ .
\ee
\end{rem}
\subsection{Minimal subtraction and the Forest Formula}\label{sec:MSandFF}

We start with a $1$-dimensional toy example, which is taken from
\cite[sect.~III.3.2]{Hoer03}, but treated here in the somewhat different light of extension of
(homogeneous) distributions from $\Dcal'(\RR\setminus\{0\})$ to $\Dcal'(\RR)$, 
cf.~\cite{NST11} and \cite{GraciaBondiaLazzarini2003}.
 
{\small\begin{example}[Toy model]\label{thm:toy} The distribution
\be
t^\zeta(x)=\Theta(x)\,x^{-k+\zeta}\in\Dcal'(\RR\setminus\{0\})\ ,\quad
k\in\NN\,,\quad\zeta\in\CC\,,\,\,|\zeta|<1\ ,
\ee
($\Theta(x)$ denotes the Heaviside function) scales homogeneously with degree $\kappa=k-\zeta$. 
We are searching almost homogeneous extensions to $\Dcal'(\RR)$, in particular for $\zeta =0$.
For $\zeta\not= 0$ the unique homogeneous extension $\dot t^\zeta\in\Dcal'(\RR)$ can 
be obtained by our formula \eqref{diffrenhom0}: the definition \eqref{omega}
gives $\omega=k-1$ and with that we obtain
\be\label{diffren-toy}
\dot t^\zeta(x)=\frac{1}{\zeta(\zeta-1)...(\zeta-k+1)}\,\frac{d^k}{dx^k}(\Theta(x)\,x^{\zeta})=:
\sum_{l=-1}^\infty \dot t_l(x)\,\zeta^l\ .
\ee
This is an analytic regularization of $t^0=\Theta(x)\,x^{-k}$ in the sense of definition 
\ref{df:regularisation}, since one verifies straightforwardly that
$$
\lim_{\zeta\to 0}\langle \dot t^\zeta , g\rangle=\int dx\,\Theta(x)\,x^{-k}\,g(x)\ ,\quad\quad\forall
g\in \Dcal_{k-1}(\RR)\ ,
$$
by using that such a function $g$ is of the form $g(x)=x^k\,\tilde g(x)$ with $\tilde g\in  \Dcal(\RR)$.
For $\zeta=0$ the extension is non-unique, almost homogeneous scaling is 
compatible with the addition of a term $C\,\delta^{(k-1)}(x)\ $, where $C\in\CC$ arbitrary.
However, the $MS$-prescription \eqref{def:MS} yields a unique result: $t^{\rm MS}(x)=\dot t_0(x)$
(the coefficient $l=0$ in the expansion \eqref{diffren-toy}). Using
\be
\Theta(x)\,x^{\zeta}=\Theta(x)+\zeta\,\Theta(x)\,\ln\,x+{\cal O}(\zeta^2)
\ee
and
\be
\frac{1}{(\zeta-1)...(\zeta-k+1)}=\frac{(-1)^{k-1}}{(k-1)!}\,
\Bigl(1+\zeta\,\sum_{j=1}^{k-1}1/j+{\cal O}(\zeta^2)\Bigr)
\ee
we obtain
\be
t^{\rm MS}(x)=\frac{(-1)^{k-1}}{(k-1)!}\,
\Bigl(\frac{d^k}{dx^k}(\Theta(x)\,\ln x)+(\sum_{j=1}^{k-1}1/j)
\,\,\delta^{(k-1)}(x)\Bigr)\ .
\ee
Note that $t^{\rm MS}$ scales almost homogeneously with degree $k$ and power $1$.
\end{example}}

We now apply the formula \eqref{diffrenhom0} to the
distributions ${t_{H,\Gamma}^{\bzeta,\beta}}\in\Dcal'(\RR^{d(n-1)}\setminus\{ 0\})$,   
arising as numerical coefficients in the expansions \eqref{causWick} of $T_{H,\Gamma}^{\bzeta}$. For such objects  $\kappa\equiv\kappa^{\bzeta,\beta}$ is given by equation \eqref{kappa:zeta} and the domain for $\bzeta\in\CC^N$
is restricted by $\Re(\kappa^{\bzeta,\beta})<\kappa^{{\bf 0},\beta}+1$ and 
$\kappa^{\bzeta,\beta}\not\in d(n-1)+\NN_0$ to the region
\be\label{Omega-graph}
\Omega_{\Gamma}^{\beta}:=\bigl\{\bzeta=(\zeta_{ij})_{1\leq i<j\leq n}\,\vert\,2\,{\bf l}\bzeta\not\in
\{0,1,...,\kappa^{{\bf 0},\beta}-d(n-1)\}\,\, \wedge\,\,\Re\, 2{\bf l}\bzeta> -1\bigr\}\,.
\ee
The minimal $\omega\in\ZZ$ satisfying \eqref{omega} for all $\bzeta$ 
fulfilling these restrictions is  
\be\label{sing-order}
\omega=\kappa^{{\bf 0},\beta}-d(n-1)\ .
\ee
Since for $\omega <0$ the direct extension (Thm.~\ref{thm:Extension-sd})
is applicable, we only study the case $\omega\geq 0$. The unique homogeneous
extension \eqref{diffrenhom0} can be written as
\be\label{diffrenhom1}
\dot{t}_{H,\Gamma}^{\bzeta,\beta}=\frac{1}{\prod_{k=0}^\omega (2{\bf l}\bzeta-k)}\,
\sum_{r_1...r_{\omega+1}}P_{r_{\omega+1}}...P_{r_{1}}
\Bigl(\overline{Q_{r_1}\dots Q_{r_{\omega+1}}\,{t_{H,\Gamma}^{\bzeta,\beta}}}\Bigr)\ .
\ee
We explicitly see that $\dot{t}_{H,\Gamma}^{\bzeta,\beta}$ has possible poles at $2\,{\bf l}\bzeta\in\{0,1,...,\omega\}$. Before we can perform the minimal subtraction, we have to pass from $T_{H,n}^\bzeta$ to  $T_{n}^\bzeta$. On the level of graphs, this corresponds to multiplying extended regularized expressions  constructed above with powers of $C^\zeta$. Since $C^\zeta$ is regular in $x$, these powers and multiplications
are well defined and no extension is needed. Let us fix a graph $\Gamma$ and consider subgraphs $\gamma,\gamma^c$ with the same vertex set such the edges of $\Gamma$ are either edges of $\gamma$ or $\gamma^c$, i.e. $E(\gamma)\subset E(\Gamma)$ and $E(\gamma^c)=E(\Gamma)\setminus E(\gamma^c)$. According to \eqref{TH:TF}, $T_{\Gamma}^{\bzeta}$ can be constructed as
\[
T_\Gamma^{\bzeta}=\sum\limits_{\gamma}T_{H,\gamma}^{\bzeta_\gamma}\circ T_{C,\gamma^c}^{\bzeta_{\gamma^c}}\,,
\]
where $\bzeta=(\bzeta_{\gamma}, \bzeta_{\gamma^c})$ and 
$T_{C,\gamma^c}^{\bzeta_{\gamma^c}}:=\prod_{i,j\in V(\gamma)}(C^{\zeta_{ij}})^{l_{ij}^{\gamma^c}}$
($l_{ij}^{\gamma^c}$ denotes the number of $(ij)$-lines in $\gamma^c$). 
Using the field expansion \eqref{causWick} we can write the corresponding formula also on the level of numerical distributions:
\be\label{zeta:expansion}
\dot{t}_\Gamma^{\bzeta,\beta}=\sum\limits_{\gamma,\beta_1\le\beta}\dot{t}_{H,\gamma}^{\bzeta_\gamma,\beta_1}\,t_{C,\gamma^c}^{\bzeta_{\gamma^c},\beta-\beta_1}\ .
\ee 
Note that $\beta_1$, contrary to $\beta$, is not restricted by the condition that it involves only the relative coordinates at each vertex.

To perform the minimal subtraction scheme we set all $\zeta_{ij}$ to be equal to a fixed value $\zeta$ and 
determine the coefficient of $\zeta^0$ in the Laurent series
\eqref{zeta:expansion}. In the massless case the  minimal subtraction scheme simplifies significantly, since we don't have to separate the regularized time-ordered products into $T^{\bzeta}_{H,\Gamma}$ and the powers of $C$.  

We illustrate the massless case for a graph $\Gamma$ with no subdivergences.  
Then the non-extended distribution $t_\Gamma^{\zeta,\beta}\equiv t^\zeta$ is analytic in $\zeta$, 
\be\label{t-analyt}
t^\zeta=t_0+\zeta\,t_1+\mathcal{O}(\zeta^2)\ ,
\ee
Thus the
extension $\dot t^\zeta$ \eqref{diffrenhom1}
has a pole of order $1$ at $\zeta =0$, i.e. $\dot t^\zeta=\sum_{j=-1}^\infty \dot t_j\,
\zeta^j\ $; the coefficient $\dot t_0$ is the $MS$-solution $t^{\rm MS}$. 
By expanding also the prefactor $\tfrac{1}{\prod_k(2c\zeta-k)}$ with $c=\sum_{i<j}l_{ij}=|E(\Gamma)|$, we obtain
\be\label{diffrenhomMS}
t^{\rm MS}=\frac{(-1)^\omega}{\omega!}
\sum_{r_1...}\partial_{r_1}...\partial_{r_{\omega+1}}
\Bigl(\frac{1}{2c}\,\overline{Q_{r_1} \dots Q_{r_{\omega+1}}\,t_1}
+(\sum_{k=1}^\omega\frac{1}{k})\,\overline{Q_{r_1} \dots Q_{r_{\omega+1}}\,t_0}\Bigr)\ .
\ee
The second term on the right hand side is a finite renormalization term, i.e.~it is of the form $\sum_{|a|=\omega} C_a\,
\partial^a\delta$ with $C_a\in\CC$, since  it vanishes for 
$x\not= 0$ due to $\partial_{r_{\omega+1}}\circ Q_{r_{\omega+1}}(Q_{r_1}\dots Q_{r_\omega}\,t_0)=0$ 
(cf.~\eqref{scaling1}).  
 
The first term in \eqref{diffrenhomMS} contains generically logarithmic terms
which come from the expansion of a product of massless Feynman propagators:
\be\label{expansion-Feyprop}
\prod_j\frac{\mu^{2\zeta}}{(-(x_j^2-i\epsilon))^{1-\zeta}}=
\frac{1+\zeta\,\sum_j\ln(-\mu^2(x_j^2-i\epsilon))+\mathcal{O}(\zeta^2)}{\prod_j(-(x_j^2-i\epsilon))}\ .
\ee

The dimensional regularization which we introduce doesn't yield 
${T^{\bzeta,\mathrm{unren}}_n}$ \eqref{Sreg-unren} (i.e.~the Feynman rules) 
finite, in the sense that the formal expressions for a graph $\Gamma$ and a multi-index $\beta$ characterizing the derivatives are a priori meaningful only for values of the regularization parameters $\bzeta$ with $\mathrm{Re}\,\zeta_{ij}$ sufficiently large.\footnote{\label{dim-reg-usual} This corresponds
to the fact that in the dimensional regularization in Euclidean momentum space the ``Feynman
integrals in $d-2\zeta$ dimensions'' are only  defined for $\mathrm{Re}\,\zeta$ sufficiently large and have to be extended by analytic continuation.}
The analytic extension to $\Omega_{\Gamma}^{\beta}$ can be constructed by the use of the homogeneous scaling with non-integer degree in terms of
formula \eqref{diffrenhom1}. 
In the presence of divergent subdiagrams, one first has to perform the analytic extension for the subdiagrams.
This amounts to solving the EG induction scheme for the deformed theory.  
The result is unique. Then, the limit $\bzeta\to {\bf 0}$ is performed by
applying the EG Forest Formula \eqref{EGforest}. A disadvantage is that in intermediate steps of the construction of the analytic extension partitions of unity are used (see Example \ref{thm:doubletriangle}) which make the method less explicit. 
An alternative is the so-called splitting method originally used by Epstein and Glaser which avoids partitions of unity on the price of a more complicated combinatorics.
\begin{rem}[Quick computation in the massive case]\label{m-quick}
If $t^\zeta_m:=t^{\zeta,\beta}_{\Gamma}$ is a product of derivated regularized Feynman propagators
(i.e.~all preceding inductive steps of the EG-construction of $t^\zeta_m$ are done by 
the direct extension), the unique $\dot t^\zeta_m$ and the unique $t^\MS_m$ can be computed by the 
following procedure, which is usually much faster than the method explained above. In this case
it suffices to work with one $\zeta$.
\begin{itemize}
\item[(1)] Insert the expansion
\be\label{reg-Feyn-prop}
\Delta^\zeta_{F,\,m}(x)=\sum_{l=0}^\infty h_l^\zeta\,\mu^{2\zeta}\,m^{2l}\,(-(x^2-i\epsilon))^{l+1-\tfrac{d}2+\zeta}
+\sum_{l=0}^\infty c_l^\zeta\,\mu^{2\zeta}\, m^{d-2+2l-2\zeta}\,(-x^2)^l
\ee
(with coefficients $h_l^\zeta,\,c_l^\zeta$ which do not depend on $(x,m)$, see \eqref{Creg}-\eqref{Hreg})
into $t^\zeta_m=\prod\partial^a\Delta^\zeta_{F,\,m}$.
\item[(2)] Let $\zeta\in\Omega^\beta_\Gamma$. Write $t^\zeta_m$ as the summands with scaling degree $> d(|V(\Gamma)|-1)-1$ 
and a remainder $r^\zeta_m$:
\be
t^\zeta_m(x)=\sum_{p=0}^P\sum_c m^{p-2c\zeta}\,\tau_{p,c}^\zeta(x)+r^\zeta_m(x)\ ,
\ee
where $c$ is the total number of $c$-lines (i.e.~the propagator is a $c_l^\zeta$-term).
\item[(3)] Apply the direct extension to $r^\zeta_m$. Since $\tau_{p,c}^\zeta(x)$ is homogeneous in $x$
with a non-integer degree, it can be extended by the differential renormalization formula  \eqref{diffrenhom1}
with $Q_0\equiv 0\equiv P_0$. Summing up we obtain
\be
\dot t^\zeta_m(x)=\sum_{p=0}^P\sum_c m^{p-2c\zeta}\,\dot\tau_{p,c}^\zeta(x)
+\overline{r^\zeta_m}(x)\ .
\ee
Obviously, the so constructed $\dot t^\zeta_m$ is the unique solution of our axioms.
\item[(4)] Minimal subtraction acts only on the expressions $(m^{-2c\zeta}\,\dot\tau_{p,c}^\zeta)$, i.e.
\be
t^\MS_m(x)=\sum_{p=0}^P m^p \sum_c \,\lim_{\zeta\to 0}\rp\bigl(m^{-2c\zeta}\,\dot\tau_{p,c}^\zeta(x)\bigr)
+\overline{r_m}(x)\ ,\quad \overline{r_m}:=\lim_{\zeta\to 0}\overline{r^\zeta_m}\ ,
\ee
because $\overline{r^\zeta_m}$ is analytic in $\zeta$. The latter can be seen as follows: for $x\not= 0$
we conclude from the analyticity of $\Delta^\zeta_{F,\,m}$ that $t^\zeta_m$ and the sums 
$\sum_c m^{-2c\zeta}\,\tau_{p,c}^\zeta$ are analytic, hence, $r^\zeta_m$ is analytic and this property is maintained 
in the direct extension.
\end{itemize}
\end{rem}

\subsection{Examples}\label{examples}
In this section we use the shorthand notations
\be
x_{kl}:=x_k-x_l\ ,\quad X:=-(x^2-i\epsilon)\ ,\quad X_{kl}:=-(x_{kl}^2-i\epsilon)\ ,
\quad X-Y:=-((x-y)^2-i\epsilon)\ .
\ee

To simplify the formulas we work with a slight modification of the 
regularized Feynman propagator: writing the prefactor $c(d-2\zeta)$ 
(used in \eqref{Creg}-\eqref{Hreg}) as 
$$
c(d-2\zeta)=\tfrac{(2\pi)^{\zeta-\frac{d}{2}}}2\,\Gamma(\tfrac{d}{2}-1-\zeta)\Gamma(2-\tfrac{d}{2}+\zeta)
$$
(by means of Euler's reflection formula),
we replace $\pi^\zeta\,\Gamma(\tfrac{d}2-1-\zeta)$ by $\Gamma(\tfrac{d}2-1)$. This amounts to a finite renormalization of 
$t^\MS$, which is analogous to the step from the $MS$- to the
$\overline{MS}$-prescription in conventional dimensional regularization.

\begin{example}[Second order of a massless model in $d=4$ dimensions]\footnote{In 
\cite{GraciaBondiaLazzarini2003} this example is treated by essentially the same method under the 
name 'analytical regularization'.} 
The $k$-th power of the dimensionally regularized massless Feynman propagator,
\be\label{DF4}
t^\zeta(x):=(D_F^\zeta(x))^k\ ,\quad k\in\{2,3,4,...\}\,,\quad\text{with}\quad
D_F^\zeta(x)=\frac{\mu^{2\zeta}}{4\pi^2\,X^{1-\zeta}}\ ,
\ee
exists in $\Dcal'(\RR^4)$ (by the direct 
extension, Thm.~\ref{thm:Extension-sd}) for 
\be
\sd(t^\zeta)<4\quad\text{that is for}\quad \Re(\zeta)>1-\tfrac{2}{k}\ .
\ee
Analytic continuation to a function meromorphic in 
$\Omega:=\{\zeta\in\CC\,|\,\Re(\zeta)>-\tfrac{1}{k}\}$
can be done by differential renormalization. Instead of using \eqref{diffrenhom1}
we proceed in the
following way:\footnote{For $k=2$ (fish diagram)
the two procedures give essentially the same formula, due to 
$\square X^\alpha=-2\alpha\,\partial_\mu(x^\mu\,
X^{\alpha-1})$; but for higher powers of $D_F^\zeta$, the method \eqref{DFk}
yields shorter formulas.}
on $\Dcal(\RR^4\setminus\{0\})$ the distribution $X^\alpha$ is
well-defined for $\alpha\in\CC$ and one easily verifies
\be
\square X^\alpha=-4\alpha(\alpha+1)\,X^{\alpha-1}
\ee
(cf.~\cite{GraciaBondiaLazzarini2003}). Hence, in $\Dcal'(\RR^4\setminus\{0\})$
$t^\zeta$ agrees with  
\be
\dot t^\zeta(x)=\frac{(-1)^{k-1}\,\mu^{2k\zeta}}
{(4\pi^2)^k\,4^{k-1}\,k\zeta\,(k\zeta-k+1)\,\prod_{j=1}^{k-2}(k\zeta-j)^2}
\,\square^{(k-1)} X^{-1+k\zeta}\ ,\label{DFk}
\ee
for almost all $\zeta\in\CC$, where $\prod_{j=1}^{k-2}(k\zeta-j)^2:=1$ for $k=2$.

As distributions on $\Dcal(\RR^4)$ we have $t^\zeta=\dot t^\zeta$ for 
$\Re(\zeta)>1-\tfrac{2}{k}$ only; however, $\dot t^\zeta$ is well-defined as meromorphic function
on $\Omega$ by direct extension of $X^{-1+k\zeta}$, 
since $\mathrm{sd}(X^{-1+k\zeta})=2(1-\Re(k\zeta))<4=d$ for $\zeta\in\Omega$.
Hence, $\dot t^\zeta$ is the unique analytic continuation of $t^\zeta$.

The $MS$-solution can be computed as explained in \eqref{t-analyt}-\eqref{expansion-Feyprop}:
\begin{align}\label{DFk-MS}
t^{\rm MS}(x)&=\frac{(-1)^{k-1}}{(4\pi^2)^k\,4^{k-1}\,(1-k)\,\prod_{j=1}^{k-2}j^2}\,\square^{(k-1)}
\Bigl(\frac{\ln(\mu^2\,X)+c}{X}\Bigr)\notag\\
&=\frac{(-1)^{k-1}}{(4\pi^2)^k\,4^{k-1}\,(1-k)\,\prod_{j=1}^{k-2}j^2}\,
\Bigl(\square^{(k-1)}\Bigl(\frac{\ln(\mu^2\,X)}
{X}\Bigr)+c\,i4\pi^2\,\square^{(k-2)}\delta(x)\Bigr)\ ,
\end{align}
where $c:=(k-1)^{-1}+2\sum_{j=1}^{k-2}j^{-1}$ and $\sum_{j=1}^{k-2}j^{-1}:=0$ for $k=2$.
\end{example}

\begin{example}[Massive setting sun diagram in $d=4$ dimensions]\label{exp:settingsun}
We use the quick computation of remark \ref{m-quick}: we write $ t^\zeta_m:=(\De_{F,m}^\zeta)^3$ as
\be\label{sm-settingsun}
t^\zeta_m(x)=\tau_{0,0}^\zeta(x)+(\tfrac{m}{\mu})^2\Bigl(\tau_{2,0}^\zeta(x)+
(\tfrac{m}{\mu})^{-2\zeta}\,\tau_{2,1}^\zeta(x)\Bigr)+r^\zeta_m(x)
\ee
with
\begin{align}
\tau_{0,0}^\zeta(x)=&(h_0)^3\,\mu^{6\zeta}\,X^{-3+3\zeta}\ ,\quad
\tau_{2,0}^\zeta(x)=3\,h_1^\zeta\,(h_0)^2\,\mu^{2+6\zeta}\,X^{-2+3\zeta}\ ,\notag\\
\tau_{2,1}^\zeta(x)=&3\,c_0^\zeta\,(h_0)^2\,\mu^{2+4\zeta}\,X^{-2+2\zeta}\ ,
\end{align}
where 
\begin{align}\label{coefficients}
h_0 &=\frac{1}{4\,\pi^2}\,,\quad
h_1^\zeta=\frac{\Gamma(\zeta)}{16\,\pi^2\,\Gamma(1+\zeta)}
=\frac{1}{16\,\pi^2}\,\Bigl(\frac{1}{\zeta}+\mathcal{O}(\zeta^0)\Bigr)\ ,\notag\\
c_0^\zeta&=-\frac{4^\zeta\,\Gamma(\zeta)}{16\,\pi^2\,\Gamma(2-\zeta)}
=\frac{-1}{16\,\pi^2}\,\Bigl(\frac{1}{\zeta}+\mathcal{O}(\zeta^0)\Bigr)\ .
\end{align}
We point out that the singularities of $h_1^\zeta,\, c_0^\zeta$ for $\zeta\to 0$
cancel out in the combination $t^\zeta_{2}:=(\tau_{2,0}^\zeta+(\tfrac{m}{\mu})^{-2\zeta}\,\tau_{2,1}^\zeta)$:
\be\label{cancellation}
\lim_{\zeta\to 0}t^\zeta_{2}(x)=3\,(h_0)^2\,\mu^2\,X^{-2}\,
\lim_{\zeta\to 0}(h_1^\zeta\,(\mu^2X)^{\zeta}+c_0^\zeta\,(\tfrac{m}{\mu})^{-2\zeta})=
\tfrac{3}{2^8\,\pi^6}\,\mu^2\,X^{-2}
\,\Bigl(\ln(\tfrac{m^2\,X}{4})-2\,\Gamma '(1)-1\Bigr)
\ee
in $\Dcal'(\RR^3\setminus\{0\})$, as it must be since $(\De_F^\zeta)^3$ is analytic in $\zeta$. 

We have to compute the coefficient
$t^\MS=\dot t_0$ of $\dot t^\zeta=\sum_{k=-l}^\infty \dot t_k\,\zeta^k$ for $t^\zeta=\tau^\zeta_{0,0}\ $,
$t^\zeta=\tau^\zeta_{2,0}$ and $t^\zeta=m^{-2\zeta}\,\tau^\zeta_{2,1}\ $. 
For the former the result is the particular case $k=3$ of \eqref{DFk-MS}:
\be
\tau^\MS_{0,0}(x)=\frac{-1}{2^{11}\,\pi^6}\,
\Bigl(\square^2\Bigl(\frac{\ln(\mu^2\,X)}
{X}\Bigr)+\tfrac{5}2\,i\,4\,\pi^2\,\square\delta(x)\Bigr)\ .
\ee
The extensions of $\tau^\zeta_{2,0}$ and $\tau^\zeta_{2,1}$ are obtained analogously to \eqref{DFk}:
\be 
\dot \tau^\zeta_{2,0}(x)=\frac{\mu^{2+6\zeta}\,3\,h_1^\zeta\,(h_0)^2}{4\,(1-3\zeta)\,3\zeta}\,
\square X^{-1+3\zeta}\ ,\quad
\dot \tau^\zeta_{2,1}(x)=\frac{\mu^{2+4\zeta}\,3\,c_0^\zeta\,(h_0)^2}{4\,(1-2\zeta)\,2\zeta}\,
\square X^{-1+2\zeta}\ .
\ee
Note that the cancellation of the leading negative power of $\zeta$ in the sum
$t^\zeta_{2}$ \eqref{cancellation}
does not work for the pertinent extensions:
\be
\dot t^\zeta_{2}:=(\dot\tau_{2,0}^\zeta+(\tfrac{m}{\mu})^{-2\zeta}\,\dot \tau_{2,1}^\zeta)
=\frac{3\,(h_0)^2}{4\,\zeta}\,\square \Bigl(X^{-1+2\zeta}\,\mu^{2+4\zeta}\Bigl[
\frac{h_1^\zeta\,(\mu^2 X)^\zeta}{(1-3\zeta)\,3}+\frac{c_0^\zeta\,(\tfrac{m}{\mu})^{-2\zeta}}{(1-2\zeta)\,2}\Bigr]\Bigr)\ ,
\ee
since $[...]=\tfrac{1}{16\,\pi^2\,\,\zeta}\,(\tfrac{1}3-\tfrac{1}2)+\mathcal{O}(\zeta^0)$. Expanding
\begin{align}
\frac{3\,h_1^\zeta\,(h_0)^2}{4\,(1-3\zeta)\,3\zeta}&=\Bigl(\frac{r_{-2}}{\zeta^2}+
\frac{r_{-1}}{\zeta}+r_0+\mathcal{O}(\zeta)\Bigr)\ ,\notag\\
\frac{3\,c_0^\zeta\,(h_0)^2}{4\,(1-2\zeta)\,2\zeta}&=\Bigl(\frac{s_{-2}}{\zeta^2}+
\frac{s_{-1}}{\zeta}+s_0+\mathcal{O}(\zeta)\Bigr)\ ,
\end{align}
the MS-solutions read
\begin{align}
\tau^\MS_{2,0}(x):=\mu^{-2}\,\lim_{\zeta\to 0}\rp (\dot\tau_{2,0}^\zeta)(x)&=
\square\Bigl(\frac{r_0+r_{-1}\,3\,\ln(\mu^2X)+r_{-2}\,\tfrac{9}2\,
(\ln(\mu^2X))^2}{X}\Bigr)\ ,\notag\\
\tilde\tau_{2,1}^\MS(x):=\mu^{-2}\,\lim_{\zeta\to 0}\rp ((\tfrac{m}{\mu})^{-2\zeta}\,\dot\tau_{2,1}^\zeta)(x)&=
\square\Bigl(\frac{s_0+s_{-1}\,2\,\ln(\tfrac{\mu^3X}{m})+s_{-2}\,2\,
(\ln(\tfrac{\mu^3X}{m}))^2}{X}\Bigr)\ .
\end{align}
Joining together the various terms we end up with
\be
t^\MS_m(x)=\tau_{0,0}^\MS(x)+m^2\Bigl(\tau_{2,0}^\MS(x)+\tilde\tau_{2,1}^\MS(x)\Bigr)
+\overline{r_m}(x)\ .
\ee

Notice that $\overline{r_m}$ can be computed directly, without using any regularization: inserting
the $m^2$-expansion of the Feynman propagator,
\be
\De_F(x)=\frac{h_0}{X}+m^2\sum_{l=0}^\infty\Bigl(q_l\,(m^2X)^l\,\ln(m^2X)+Q_l\,(m^2X)^l\Bigr)\ ,
\quad q_l,Q_l\in\RR\ ,
\ee
into $(\De_F)^3$, $\overline{r_m}$ is the direct extension (Thm.~\ref{thm:Extension-sd}) of the sum of all 
terms which are $\sim m^{4+2l}\,(\ln(m^2X))^k$ with $l\in\NN_0$ (and $k=0,1,2,3$). 
\end{example}

\begin{example}[Massless triangle diagram in $d=6$ dimensions]\label{exp:triangle}
The dimensionally regularized massless Feynman propagator
in $d=6$ dimensions reads
\be\label{Feynman6}
D_F^\zeta(x)=\frac{\mu^{2\zeta}}{4\pi^3\,X^{2-\zeta}}\ ,
\ee
The triangle diagram
\be
t^\zeta(x,y)=D_F^\zeta(x)\,D_F^\zeta(y)\,D_F^\zeta(x-y)\,\in\Dcal'(\RR^{12}\setminus\{0\})\ ,\quad 
\zeta\in\CC\setminus\{0\}\ ,
\ee
is homogeneous with degree $\kappa^\zeta=12-6\zeta$, i.e.~we have $\omega =0$
\eqref{omega}. With that \eqref{diffrenhom1} yields
\be\label{triangle6}
\dot t^\zeta(x,y)=\frac{\mu^{6\zeta}}{4^3\pi^9\,6\,\zeta}\,
\partial_{(x,y)\mu} \Bigl(\overline{\frac{(x,y)^\mu}{X^{2-\zeta}\,
Y^{2-\zeta}\,(X-Y)^{2-\zeta}}}\Bigr)\ ,
\ee
where we write $\partial_{(x,y)\mu} \bigl(\overline{(x,y)^\mu\,t(x,y)}\bigr)$ for $\partial_{x,\mu} 
(\overline{x^\mu t(x,y)})+\partial_{y,\mu} (\overline{y^\mu t(x,y)})$ to simplify the notation.

For $\zeta\to 0$ the $MS$-prescription \eqref{diffrenhomMS} yields
\be\label{triangle6MS}
t^{\rm MS}(x,y)=\frac{1}{3\cdot 2^7\,\pi^9}\,\partial_{(x,y)\mu}
\Bigl(\overline{(x,y)^\mu\,\frac{\ln(\mu^6\,X\,Y\,(X-Y))}
{X^2\,Y^2\,(X-Y)^2}}\Bigr) 
\ee
\end{example}

\begin{example}[Massless triangle diagram with subdivergences in $d=4$ dimensions]\label{exp:triangle1}
We want to renormalize
\be
t^\bzeta(x,y)=\Bigl((D_F^{\zeta_1}(x))^2\Bigr)_\mathrm{ren}\,\Bigl((D_F^{\zeta_2}(y))^2\Bigr)_\mathrm{ren}
D_F^{\zeta_3}(x-y)\,\in\Dcal'(\RR^8\setminus\{0\})\ ,
\ee
with $\bzeta\equiv (\zeta_1,\zeta_2,\zeta_3)\in\CC^3\ ,\,\,|\bzeta|$ small enough, 
$\zeta_1\not= 0$, $\zeta_2\not= 0$, $(2\zeta_1+2\zeta_2+\zeta_3)\not= 0$
and where $D_F^\zeta$ is given by \eqref{DF4}.
By 'ren' we mean that the pertinent (divergent) subdiagram is renormalized. Doing this 
by using \eqref{DFk} we have
\be
t^\bzeta(x,y)=c(\zeta_1)\,c(\zeta_2)\,\tau^\bzeta(x,y)\ ,
\ee
where
\be
c(\zeta):=\frac{\mu^{4\zeta}}{(4\pi^2)^2\,8\,\zeta\,(1-2\zeta)}
=:\sum_{k=-1}^\infty c_k\,\zeta^k
\ee
and
\be
\tau^\bzeta(x,y):=\frac{\mu^{2\zeta_3}}{4\pi^2}
\Bigl(\square_x X^{-1+2\zeta_1}\Bigr)\,\Bigl(\square_y Y^{-1+2\zeta_2}\Bigr)\,
(X-Y)^{-1+\zeta_3}\ .
\ee
We explicitly see that $t^\bzeta$ scales homogeneously with 
degree $\kappa^\bzeta=10-2(2\zeta_1+2\zeta_2+\zeta_3)$. Hence, 
the differential renormalization formula \eqref{diffrenhom1}
can be applied with $\omega=2$:
\be
\dot t^\bzeta(x,y)=\frac{c(\zeta_1)\,c(\zeta_2)}{2\zeta_1+2\zeta_2+\zeta_3}\,
\tilde t^\bzeta(x,y)\in\Dcal^\prime(\RR^8)\ ,
\ee
where
\begin{align}\label{trianglesubsub}
\tilde t^\bzeta(x,y)&:=\frac{1}{2\,(1-4\zeta_1-4\zeta_2-2\zeta_3)\,(2-4\zeta_1-4\zeta_2-2\zeta_3)}\Bigl(
\partial_{x,\mu}\partial_{x,\nu}\partial_{x,\lambda}\bigl(\overline{x^\mu x^\nu x^\lambda \tau^\bzeta(x,y)}\bigr)
\notag\\
&+3\,\partial_{x,\mu}\partial_{x,\nu}\partial_{y,\lambda}\bigl(\overline{x^\mu x^\nu y^\lambda \tau^\bzeta(x,y)}\bigr)
+3\,\partial_{x,\mu}\partial_{y,\nu}\partial_{y,\lambda}\bigl(\overline{x^\mu y^\nu y^\lambda \tau^\bzeta(x,y)}\bigr)
\notag\\
&+\partial_{y,\mu}\partial_{y,\nu}\partial_{y,\lambda}\bigl(\overline{y^\mu y^\nu y^\lambda \tau^\bzeta(x,y)}\bigr)\Bigr)
\end{align}
is analytic in $\bzeta$ for $|\bzeta|$ sufficiently small. 

Turning to the limit $\bzeta\to {\bf 0}$ we apply the EG Forest Formula \eqref{EGforest} (Thm.~\ref{forest}):
first we subtract the principle parts of the divergent subdiagrams
\begin{align}\label{ssd}
(1+R_{\zeta_1}+R_{\zeta_2})\,\dot t^\bzeta(x,y)&=\frac{c(\zeta_1)\,c(\zeta_2)}{2\zeta_1+2\zeta_2+\zeta_3}\,
\tilde t^{(\zeta_1,\zeta_2,\zeta_3)}(x,y)\notag\\
&-\frac{c_{-1}\,c(\zeta_2)}{\zeta_1\,(2\zeta_2+\zeta_3)}\tilde t^{(0,\zeta_2,\zeta_3)}(x,y)
-\frac{c_{-1}\,c(\zeta_1)}{\zeta_2\,(2\zeta_1+\zeta_3)}\tilde t^{(\zeta_1,0,\zeta_3)}(x,y)
\end{align}
(where we write $R_{\Lambda_I}$ instead of $R_I$). Using $\square_x X^{-1}=-i 4 \pi^2\,\delta(x)$
we explicitly see that $\tilde t^{(0,\zeta_2,\zeta_3)}(x,y)= \tilde t^{\zeta_2,0,\zeta_3)}(y,x)$ 
has support on the partial diagonal $x=0$:
\begin{align} 
\tilde t^{(0,\zeta_2,\zeta_3)}(x,y)=& \frac{-i\,\mu^{2\zeta_3}}
{2\,(1-4\zeta_2-2\zeta_3)\,(2-4\zeta_2-2\zeta_3)}\,\delta(x)\notag\\ 
&\quad\cdot \partial_{y,\mu}\partial_{y,\nu}\partial_{y,\lambda}\Bigl(\overline{y^\mu y^\nu y^\lambda 
\bigl(\square_y Y^{-1+2\zeta_2}\bigr)\, 
Y^{-1+\zeta_3}}\Bigr)\ .  
\end{align}
To obtain 
\be
t^\mathrm{MS}(x,y)=\lim_{\zeta\to 0}\,(1+R_\zeta)\vert_{\zeta:=\zeta_1=\zeta_2=\zeta_3}
(1+R_{\zeta_1}+R_{\zeta_2})\,\dot t^\bzeta(x,y)
\ee
($(1+R_\zeta)$ removes the remaining ``overall divergence'') we set $\zeta:=\zeta_1=\zeta_2=\zeta_3$
in $(1+R_{\zeta_1}+R_{\zeta_2})\,\dot t^\bzeta(x,y)$ \eqref{ssd} and compute from the resulting 
Laurent series in $\zeta$ (which has a pole of order $3$) the term $\sim\zeta^0$. Using
the expansions
\be
\tilde t^{(\zeta,\zeta,\zeta)}(x,y)=\sum_{k=0}^\infty t_k(x,y)\,\zeta^k\ ,\quad\quad
\tilde t^{(0,\zeta,\zeta)}(x,y)=\sum_{k=0}^\infty t^1_k(x,y)\,\zeta^k\ ,
\ee
we end up with
\be
t^\mathrm{MS}(x,y)=\frac{1}5\,\sum_{q+r+s=1}c_q\,c_r\,t_s(x,y)-
\frac{c_{-1}}3\,\sum_{r+s=2}c_r\,(t^1_s(x,y)+t^1_s(y,x))\ ,
\ee
where $q,r\geq -1$ and $s\geq 0$.
\end{example}

\begin{example}[Massless double triangle diagram with overlapping divergences in $d=6$ 
dimensions]\label{thm:doubletriangle}
We introduce the notation $D_F^\zeta(x)=: d(\zeta)\, X^{-2+\zeta}$
for the Feynman propagator \eqref{Feynman6}, note that $\zeta\mapsto d(\zeta)$ is analytic. 
Compared with the preceding examples, the additional complication
of the double triangle diagram,
\be
t_\mathrm{unren}(x_{14},x_{24},x_{34})=D_F^{\zeta_{12}}(x_{12})\,D_F^{\zeta_{13}}(x_{13})\,
D_F^{\zeta_{23}}(x_{23})\,D_F^{\zeta_{24}}(x_{24})\,D_F^{\zeta_{34}}(x_{34})\ ,
\ee
is an ``overlapping divergence''. The subdiagram 123 (i.e.~with vertices $x_1,x_2,x_3$) is computed in 
Example \ref{exp:triangle}; with different $\zeta_{kl}$'s the regularized amplitude \eqref{triangle6}
reads
\be\label{subtriangle}
\dot t_3^{(\zeta_{12},\zeta_{13},\zeta_{23})}(x_{12},x_{13})=
\frac{d(\zeta_{12})d(\zeta_{13})d(\zeta_{23})}{2(\zeta_{12}+\zeta_{13}+\zeta_{23})}\,\,
\tilde t^{(\zeta_{12},\zeta_{13},\zeta_{23})}(x_{12},x_{13})\in\Dcal^\prime(\RR^{12})\ ,
\ee
where
\be
\tilde t^{(\zeta_{12},\zeta_{13},\zeta_{23})}(x_{12},x_{13}):=
\partial_{(x_{12},x_{13})\mu} \Bigl(\overline{\frac{(x_{12},x_{13})^\mu}{
X_{12}^{\>\>2-\zeta_{12}}\,X_{13}^{\>\>2-\zeta_{13}}\,
X_{23}^{\>\>2-\zeta_{23}}}}\Bigr)
\ee
is analytic in $(\zeta_{12},\zeta_{13},\zeta_{23})$.
The second divergent subdiagram, which is 234, is obtained from 123 by replacing 1 by 4.

The whole diagram 1234 has $\omega =2$; we use the notations ${\bf x}:=(x_{12},x_{13},x_{14})$
and $\bzeta:=(\zeta_{12},\zeta_{13},\zeta_{23},\zeta_{24},\zeta_{34})$. To write down $t^\bzeta({\bf x})
\in\Dcal^\prime(\RR^{18}\setminus\{0\})$ we need to introduce a partition of unity: for ${\bf x}
\in\RR^{18}\setminus\{0\}$ let
\be\label{partofunity1}
1=f_1({\bf x})+ f_2({\bf x}) \quad\text{with}\quad f_1,f_2\in {\cal C}^\infty (\RR^{18}\setminus\{0\})
\ee
and
\be\label{partofunity2}
\supp\,f_1\subset \RR^{18}\setminus\{{\bf x}\,\vert\,x_{24}=0=x_{34}\}\ ,\quad
\supp\,f_2\subset \RR^{18}\setminus\{{\bf x}\,\vert\,x_{12}=0=x_{13}\}\ .
\ee
With that we can write
\be\label{dt1}
t^\bzeta({\bf x})=\frac{\prod d(\zeta_{kl})}{2}\,\Bigl(f_1({\bf x})\,
\frac{t^\bzeta_{123|4}({\bf x})}{\zeta_{12}+\zeta_{13}+\zeta_{23}}+
f_2({\bf x})\,\frac{(1\leftrightarrow 4)}{\zeta_{24}+\zeta_{34}+\zeta_{23}}\Bigr)
\in\Dcal^\prime(\RR^{18}\setminus\{0\})\ ,
\ee
where
\be
t^\bzeta_{123|4}({\bf x}):=\tilde t^{(\zeta_{12},\zeta_{13},\zeta_{23})}(x_{12},x_{13})\,
X_{24}^{\>\>-2+\zeta_{24}}\,X_{34}^{\>\>-2+\zeta_{34}}
\ee
is analytic in $\bzeta$. Here and in the following we mean by $\prod$ and $\sum$ the
product or sum over $(k,l)=(1,2),\,(1,3),\,(2,3),\,(2,4),\,(3,4)$. 

We point out that
$t^\bzeta$ \eqref{dt1} is {\it independent of the choice of $f_1,\,f_2$}, because on
\be
\Bigl(\supp\,f_1\cap\supp\,f_2\Bigr)\subset\Bigl(\{{\bf x}\,|\,x_{23}\not= 0\}\cup
\{{\bf x}\,|\,x_{23}= 0\,\,\wedge\,\,x_{12}\not= 0\not= x_{34}\}\Bigr)
\ee
the distribution $\dot t_3^{(\zeta_{12},\zeta_{13},\zeta_{23})}(x_{12},x_{13})$ \eqref{subtriangle}
is equal to its non-extended version and, hence,
\be\label{independent}
\frac{t^\bzeta_{123|4}({\bf x})}{2(\zeta_{12}+\zeta_{13}+\zeta_{23})}=
\prod X_{kl}^{\>\>-2+\zeta_{kl}}\ ,
\ee
is invariant under $(1\leftrightarrow 4)$. 

The differential renormalization formula \eqref{diffrenhom1}
(with $\omega =2$) yields the extension
\begin{align}\label{dt2}
\dot t^\bzeta({\bf x})=&\frac{\prod d(\zeta_{kl})}
{(\sum\zeta_{kl})\,4\,(1-2\sum\zeta_{kl})(2-2\sum\zeta_{kl})}\,\,\partial_{{\bf x}\mu}
\partial_{{\bf x}\nu}\partial_{{\bf x}\lambda}\notag\\
&\quad\quad\Bigl(\overline{{\bf x}^\mu {\bf x}^\nu {\bf x}^\lambda
\Bigl(f_1({\bf x})\,\frac{t^\bzeta_{123|4}({\bf x})}{\zeta_{12}+\zeta_{13}+\zeta_{23}}+
f_2({\bf x})\,\frac{(1\leftrightarrow 4)}{\zeta_{24}+\zeta_{34}+\zeta_{23}}\Bigr)}\Bigr)
\in\Dcal^\prime(\RR^{18})\ ,
\end{align}
where similarly to \eqref{triangle6} a shorthand notation is used (the detailed version is
analogous to \eqref{trianglesubsub}). 

We use the EG Forest Formula \eqref{EGforest} (Thm.~\ref{forest}) to compute $t^\mathrm{MS}$: 
\be
t^\mathrm{MS}({\bf x})=\lim_{\bzeta\to {\bf 0}}\,(1+R_{1234})
(1+R_{123}+R_{234})\,\dot t^\bzeta({\bf x})\ .
\ee
We point out that $R_{123}$ gives a non-vanishing contribution {\it only} on the $f_1$-term; 
because, setting $\zeta:=\zeta_{12}=\zeta_{13}=\zeta_{23}$, the $f_2$-term is analytic in a 
neighbourhood of $\zeta =0$. 
Taking this into account the counter term of the 123-subdiagram reads
\begin{align}
R_{123}\,\dot t^\bzeta({\bf x})&=R_\zeta\,\frac{d(\zeta)^3\,d(\zeta_{24})\,d(\zeta_{34})}
{\zeta\cdot 12\cdot (3\zeta+\zeta_{24}+\zeta_{34})(1-2(3\zeta+\zeta_{24}+\zeta_{34}))
(2-2(3\zeta+\zeta_{24}+\zeta_{34}))}\notag\\
&\quad\quad\quad\quad\cdot \partial_{{\bf x}\mu}\partial_{{\bf x}\nu}\partial_{{\bf x}\lambda}
\Bigl(\overline{{\bf x}^\mu {\bf x}^\nu {\bf x}^\lambda\,f_1({\bf x})\,
t^{(\zeta,\zeta,\zeta,\zeta_{24},\zeta_{34})}_{123|4}({\bf x})}\Bigr)\notag\\
&=-\frac{d(0)^3\,d(\zeta_{24})\,d(\zeta_{34})\,\pi^6}
{\zeta\cdot 12\cdot (\zeta_{24}+\zeta_{34})(1-2(\zeta_{24}+\zeta_{34}))
(2-2(\zeta_{24}+\zeta_{34}))}\notag\\
&\quad\quad\quad\quad \cdot\partial_{{\bf x}\mu}\partial_{{\bf x}\nu}\partial_{{\bf x}\lambda}
\Bigl(\overline{{\bf x}^\mu {\bf x}^\nu {\bf x}^\lambda\,\delta(x_{12},x_{13})\,
X_{24}^{\>\>-2+\zeta_{24}}\,X_{34}^{\>\>-2+\zeta_{34}}}\Bigr)\ ,
\end{align}
where we use the result for the scaling anomaly of the 123-subdiagram,
$\tilde t^{(0,0,0)}(x_{12},x_{13})=\pi^6\,\delta(x_{12},x_{13})$ (derived e.g.~in \cite[sect.~7.1]{BDF09}),
and that we may replace $f_1({\bf x})$ by 1 (due to $0=f_2({\bf x})\,\delta(x_{12},x_{13})
=(1-f_1({\bf x}))\,\delta(x_{12},x_{13})$). With that we explicitly see that the {\it 123-counter term
is independent of the choice of $f_1,f_2$}. 

Finally $t^\mathrm{MS}({\bf x})$ is the coefficient 
$\sim\zeta^0$ of the Laurent series $(1+R_{123}+R_{234})\,
\dot t^\bzeta({\bf x})\vert_{\bzeta=(\zeta,\zeta,\zeta,\zeta,\zeta)}$,
i.e.~it is the coefficient $\sim\zeta^2$ of the power series
\begin{align}\label{t-MS}
&\frac{d(\zeta)^5}
{60\,(1-10\zeta)(2-10\zeta)}\,\partial_{{\bf x}\mu}
\partial_{{\bf x}\nu}\partial_{{\bf x}\lambda}\Bigl(\overline{{\bf x}^\mu {\bf x}^\nu {\bf x}^\lambda
\Bigl(f_1({\bf x})\,t^{(\zeta,\zeta,\zeta,\zeta,\zeta)}_{123|4}({\bf x})+
f_2({\bf x})\,(1\leftrightarrow 4)\Bigr)}\Bigr)\notag\\
&-\frac{d(0)^3\,d(\zeta)^2\,\pi^6}
{24\,(1-4\zeta)(2-4\zeta)}\partial_{{\bf x}\mu}\partial_{{\bf x}\nu}\partial_{{\bf x}\lambda}
\Bigl(\overline{{\bf x}^\mu {\bf x}^\nu {\bf x}^\lambda\,\bigl(\delta(x_{12},x_{13})\,
X_{24}^{\>\>-2+\zeta}\,X_{34}^{\>\>-2+\zeta}+(1\leftrightarrow 4)\bigr)}\Bigr)\ .
\end{align}
Due to \eqref{independent}, this result is independent of the choice of $f_1,f_2$.
However, to compute $t^\mathrm{MS}$, an explicit choice of $f_1$ and $f_2$ is needed.
This can be done as follows: Let $\chi$ be a smooth approximation of the Heaviside-function
$\Theta (x)$ with $\supp\, \chi'\subset [-\epsilon,\epsilon]$ for a sufficiently small
$\epsilon >0$. For $x\in\RR^6$ we mean by $|x|$ the Euclidean norm. We set
\begin{align}
g_2({\bf x}):=
\begin{cases}   
\chi\Bigl(\frac{|x_{12}|^2+|x_{13}|^2}{|x_{14}|^2}-a\Bigr)\quad & \text{if}\quad x_{14}\not= 0\\
1 \quad & \text{if}\quad x_{14}= 0\>\>\wedge\>\>{\bf x}\not={\bf 0} 
\end{cases}\ ,
\end{align}
with a sufficiently small $a >0$. Note that $g_2\in {\cal C}^\infty (\RR^{18}\setminus\{0\})$. 
Let
\be
g_1({\bf x}):=g_2(x_{24},x_{34},x_{14})\ .
\ee
The sets
\be
K_j:=\{{\bf x}\not={\bf 0}\,|\,g_j({\bf x})=0\}\ ,\quad j=1,2\ ,
\ee
are narrow cones around the diagonal $x_{12}=x_{13}=x_{14}$ (for $g_1$) and the $x_{14}$-axis (for $g_2$).
Since $K_1\cap K_2 =\emptyset$, we have $g_1({\bf x})+g_2({\bf x})>0$ for all ${\bf x}\not={\bf 0}$ and,
hence, we may set
\be
f_j({\bf x}):=\frac{g_j({\bf x})}{g_1({\bf x})+g_2({\bf x})}\ ,\quad j=1,2\ .
\ee
Obviously the pair $(f_1,f_2)$ satisfies the required properties \eqref{partofunity1} and
\eqref{partofunity2}, and it also fulfills the $(1\leftrightarrow 4)$-symmetry
and scales homogeneously with degree $0$:
\be
f_1({\bf x})=f_2(-x_{24},-x_{34},-x_{14})\ ,\quad f_j(\lambda{\bf x})=f_j({\bf x})\quad\forall\lambda\not= 0\ .
\ee
The computation of $t^\mathrm{MS}$ needs explicit formulas for $f_1,f_2$ only in the first line of
\eqref{t-MS}; to compute the latter smoothness of $f_1,f_2$ is not necessary -- $\chi$ can be replaced by $\Theta$.
\end{example}
 
The computational difficulty that, in case of overlapping divergences, our method needs explicit
formulas for a partition of unity, can be avoided by using the distribution-splitting method of
Epstein-Glaser \cite{EG73} or Steinmann's direct construction of retarded products \cite{Ste71,DF04}
(instead of Stora's extension of distributions \cite{Stora1993}).
In the splitting method, a map $\Dcal_n^\bzeta:\Floc^{\otimes n}\rightarrow\Fcal$ corresponds to
$\TT^{\bzeta}_n$ restricted to the complement of the thin diagonal 
$\Delta_n:=\{(x_1,\ldots ,x_n)\,|\,x_1=\cdots =x_n\}$. 
This $\Dcal_n^\bzeta$ is a sum of products of
time-ordered products $\TT^{\bzeta}_k$ of lower orders $k<n$, 
its construction does not need any partition of unity. 
$\Dcal_n^\bzeta$ has causal support, that is the pertinent numerical distributions 
$d(x_1-x_n,...):=d^{\bzeta\,\beta}_{\alpha}(x_1,...,x_n)$ 
(defined analogously to $t^{\bzeta\,\beta}_{\alpha}$ in \eqref{causWick})
have support in $(\bar V_+)^{\times (n-1)}\,\cup\,(\bar V_-)^{\times (n-1)}$. The distribution splitting,
\begin{align}\label{splitting0}
d=a-r\ ,\quad & \text{with}\quad \supp\, a\subset (\bar V_+)^{\times (n-1)}\quad\wedge\quad
\supp\, r\subset (\bar V_-)^{\times (n-1)}\notag\\
& \text{and}\quad \mathrm{sd}(a)\leq\mathrm{sd}(d)\quad\wedge\quad
\mathrm{sd}(r)\leq \mathrm{sd}(d)\ ,
\end{align}
corresponds to the extension of $t^{\bzeta\,\beta}_{\alpha}(x_1-x_n,...)$ from
$\Dcal^\prime(\RR^l\setminus\{0\})$ to $\Dcal^\prime(\RR^l)$ (where $l=d(n-1)$), i.e.~it
can be understood as renormalization. Our results about minimal subtraction (Sect.~\ref{MSnumerical})
and the differential renormalization formula \eqref{diffrenhom1} hold, suitably reformulated, also for the 
splitting problem \eqref{splitting0}. This is worked out in Appendix \ref{app:EG-splitting}.
A main disadvantage of the 
splitting method is that usually the computation of $\Dcal_n^\bzeta$ requires quite a lot of work
(see the examples in \cite{Sch89}), so in absence of overlapping divergences the extension method 
is mostly much more efficient.
\section{Hopf Algebra and Renormalization}
In pioneering work \cite{CK00,CK01}, Connes and Kreimer uncovered interesting algebraic structures underlying the combinatorics of perturbative renormalization. In particular, one obtains the diffeomorphism group on the space of coupling constants, tangent to the identity, which is nothing else than the 
renormalization group in the sense of St{\"u}ckelberg and Petermann and was independently derived in form of the Main Theorem of Renormalization in \cite{SP82,Pin01,DF04,BDF09}. Other structures 
refer to computational methods and can best be formulated in terms of Hopf algebras of graphs.

In this section we will describe the main combinatorial structure arising in our framework and relate it to a certain Hopf algebra. We also argue how this structure can be related to the one used by Connes and Kreimer. Let $\Rcal$ denote the St\"uckelberg Petermann renormalization group. The Main Theorem of renormalization describes the set of scattering matrices $\mathscr{S}$ as a right  group module (right action),
\be
\begin{array}{rccl}
\rho: & \mathscr{S}\otimes \Rcal & \rightarrow &\mathscr{S} \\   
      & \Scal \otimes \Zcal    & \mapsto     & \Scal':=\Scal\circ \Zcal
\end{array}
\ee
Consider the formal symbols $\delta^n$, $n\in\NN$. Using the prescription
\be
\delta^n(\Scal):=\Scal^{(n)}(0):\Floc^{\otimes n}\rightarrow\Fcal \quad\text{and}\quad
\delta^n (\Zcal):=\Zcal^{(n)}(0):\Floc^{\otimes n}\rightarrow\Floc
\ee
we associate $\delta^n$'s with maps from $\mathscr{S}$ and $\Rcal$ to the $\CC$-linear space of linear maps $\Lin(\Floc^{\otimes n},\Fcal)$ and $\Lin(\Floc^{\otimes n},\Floc)$, respectively. Using this identification we can define on symbols $\delta^n$ 
two distinguished products. We start with  the tensor product $\delta^n\otimes \delta^m$ of linear mappings, which is defined to be a map form $\Rcal$ to $\Lin(\Floc^{\otimes n+m},\Floc^{\otimes2})$. Apart from $\otimes$ it is also natural to consider another product, the non-commutative composition product $\oc$ defined for symbols $\delta^{n}$, $\delta^{k_1}\otimes\dots\otimes\delta^{k_m}$ by
\[
(\delta^{n}\oc\delta^{k_1}\otimes\dots\otimes\delta^{k_m})(\Scal,\Zcal)\doteq \left\{\begin{array}{lcl}
\Scal^{(n)}\circ (\Zcal^{(k_1)}\otimes\dots\otimes \Zcal^{(k_m)})
&,&\textrm{if}\quad m=n\,,\\
0&,&\textrm{else}\,,\end{array}\right.
\]
where $\Scal\in\mathscr{S}$, $R\in\Rcal$. Note that $\Scal^{(1)}(0)=\Zcal^{(1)}(0)=\id$, i.e. $\delta^1=1$ is the unit with respect to $\oc$. We can now write the termwise version of the main theorem (i.e.~the Fa\`a di Bruno formula 
for $(\Scal\circ \Zcal)^{(n)}$, cf.~\eqref{FaadiBruno})  and the group law of $\Rcal$ as,
\begin{align}
\delta^n(\Scal\circ \Zcal)=\sum_P(\delta^{|P|}\oc\bigotimes_{I\in P}\delta^{|I|})(\Scal,\Zcal)\label{main-Hopf}\,,\\
\delta^n(\Zcal_1\circ \Zcal_2)=\sum_P(\delta^{|P|}\oc\bigotimes_{I\in P}\delta^{|I|})(\Zcal_1,\Zcal_2)\label{main-Hopf2}\,.
\end{align}
Now we want to reinterpret these formulas in the Hopf algebraic language. To construct the Hopf algebra dual to $\Rcal$ we consider first the algebra $\Oalg$ of functions on $\Rcal$ with values in $\RR$. We want to encode the group law in the coproduct structure, i.e. we want to define $\tilde{\De}:\Oalg\rightarrow\Oalg\otimes\Oalg$ such that
\be\label{coprod}
\tilde{\De}f(\Zcal_1,\Zcal_2)=f(\Zcal_1\circ \Zcal_2)\,.
\ee
This is in general not possible, since $\Rcal$ is not finite. One can fix the problem by replacing the algebraic tensor product with some completed tensor product (see for example \cite{Far00} for the case of Hopf algebras of smooth functions) or restrict oneself to the algebra of representative functions\footnote{We recall that a function is representative if its orbit under the left translation is a finite dimensional subspace of $\Oalg$. In particular, matrix elements of finite dimensional representations are such functions. For more details see the review paper \cite{FGB05} and lecture notes \cite{Frabetti2007}.}. Fortunately the situation simplifies significantly if we take into account the fact that, as shown in \cite{DF04,BDF09}, $\Rcal$ acts on the space of actions\footnote{By actions we mean equivalence classes of generalized Lagrangians, in the sense of \eqref{equ} and the discussion above it.}. Moreover, in a renormalizable theory the orbit of the interaction can be described by a finite set of parameters (coupling constants), so there exists a group morphism map from $\Rcal$ to a subspace $\tilde{\Rcal}$ of $\mathrm{Diff}(\RR^N)$, the group of formal diffeomorphisms with $N\in\NN$. From the physical point of view, all the relevant information about the theory is contained in  $\tilde{\Rcal}$, so we can now focus our attention on this group. We consider now the algebra $\Halg$ spanned by symbols $\delta^{\al,i}$, where $\delta^{\al,i}(\Zcal)\doteq \partial_\al \Zcal(0)^i$, $\Zcal\in\tilde{\Rcal}$, $\Zcal^i$ is the i-th component of $\Zcal$ and $\al\in \NN_0^N$ is a multiindex. The group law in $\mathrm{Diff}(V)$ is the composition of  diffeomorphisms of $V$ and it is easy to check that the functions $\delta^{\al,i}$ are representative. We assumed that $\Zcal(0)=0$, $\Zcal^{(1)}(0)=\id$, for $\Zcal\in\Rcal$, so $\delta^{0,i}$ is trivial and $\delta^{j,i}:\Zcal\mapsto \partial_j\Zcal^i(0)$ is identically $1$ for $j=i$ and $0$ otherwise. Therefore we identify all $\delta^{i,i}, i=1,\dots,N$ with the unit $1$ element and define the counit by setting: $\epsilon(\delta^{\al,i}):=0$, for $\al\neq i$ and  $\epsilon(1)=1$. The coproduct of $\Halg$ is defined by \eqref{coprod} and the explicit formula is just the Fa\`a di Bruno formula for maps $\RR^N\rightarrow\RR^N$.
Next we introduce on $\Halg$ the grading $\deg(\delta^{\al,i})=|\al|-1$. With this definition, $\Halg$ is an $\NN_0$-graded connected bialgebra and from the result of \cite{Kastler2000} follows that $\Halg$ has an antipode 
and, hence, is a Hopf algebra. This way we have constructed the Hopf algebra induced by action of the renormalization group $\Rcal$ on the space of coupling constants of a given renormalizable theory.
{\small\begin{example}[Fa\`a di Bruno Hopf algebra]
Let us consider the case where the space of coupling constants is one dimensional. This happens for example in case of $\ph^3$ in 6 dimensions, after performing the wave function and mass renormalization (see example 7.1 in \cite{BDF09} for details). $\Rcal$ is mapped to 
$\tilde{\Rcal}\subset \mathrm{Diff}(\RR)$ and we consider the algebra $\Falg$ generated by functions $\delta^n$, where $\delta^n(\Zcal)=\Zcal^{(n)}(0)$, $\Zcal\in \tilde{\Rcal}$. The product is just the pointwise product of functions. The action of the St\"uckelberg-Petermann renormalization group on itself (formula \eqref{main-Hopf2})  induces a coproduct $\De:\Falg\rightarrow\Falg\otimes\Falg$ by
\be
\Delta\delta^n\doteq   \sum_P\delta^{|P|}\otimes\prod_{I\in P}\delta^{|I|}\,,
\label{coproduct}
\ee
and one can write \eqref{main-Hopf2} in the form
\[
\delta^n(\Zcal_1\circ \Zcal_2)=m\circ\Delta\delta^n (\Zcal_1,\Zcal_2)\,.
\]
$\Hcal$ is $\NN$-graded 
by the order of the derivatives,
\be
\deg(\delta^n):=n-1\,,\quad\Falg=\bigoplus_{n=0}^\infty\Falg^n\ .
\ee 
The unit of $\Falg$ is $1=\delta^1$ and from $\Zcal^{(1)}(0)=1$, $\Zcal\in\tilde\Rcal$ follows that $\Falg$ is connected. $\Falg$ also has a counit  $\epsilon:\Falg\to \CC$ by $\epsilon(\delta^{n_1}
\cdot...\cdot\delta^{n_l}):=\delta_{1n_1}\cdot...\cdot\delta_{1n_l}$ 
($\delta_{ij}$ means the Kronecker delta) and an antipode  $\Ap:\Falg\to\Falg $, which is obtained by recursion from its definition as 
\be
\Ap(1):=1\quad\text{and}\quad 0=({\rm id}*\Ap)(\delta^n):=m\circ({\rm id}\otimes\Ap)\circ
\Delta(\delta^n)\quad\text{for}\quad n>1
\ee
(where $\Ap(\bigodot_{I\in P}\delta^{|I|}):=\bigodot_{I\in P}\Ap(\delta^{|I|})$), which gives
the recursive relation
\be\label{Ap}
\Ap(\delta^n):=-\sum_{|P|>1}\delta^{|P|}\cdot\prod_{I\in P}\Ap(\delta^{|I|})
\quad\text{for}\quad n>1\ .
\ee
The resulting Hopf algebra $(\Falg,\cdot,\Delta,1,\epsilon,\Ap)$ is a well known structure called the Fa\`a di Bruno Hopf algebra \cite{JoniRota1982} (see also \cite{FGB05} in the context of renormalization).
\end{example}}
To relate our Hopf-algebraic to the Connes-Kreimer approach one has to use the expansion of $S^{(n)}$ into graphs, given by relations  \eqref{time:ord} and \eqref{GraphDO}. Let us start on the abstract level of multilinear maps between spaces of functionals. According to \eqref{GraphDO}, with a graph $\Gamma$, we associate a functional differential operator $T_\Gamma$ from $\Fcal_\loc^{\otimes V(\Gamma)}$ to $(\Fcal^{\otimes V(\Gamma)})_\loc$. The notation $\mathcal{F}_{\mathrm{loc}}^{\otimes V(\Gamma)}$ means that the factors of the tensor product are numbered by the indices of $\Gamma$, i.e. for each vertex $i$ we have a variable $\ph_i$. At the end we set all the $\ph_i$ to be equal (by applying $m_n$), but for now it is important to keep track of the information, which functional derivatives are applied at which vertex. The space $(\mathcal{F}^{\otimes V(\Gamma)})_{\mathrm{loc}}$ contains functional which are local as functions of the multiplet $(\ph_i; i\in V(\Gamma))$, i.e. depending on field configurations only through the jet $(x,\ph_i(x),\partial\ph_i(x),\dots; i\in V(\Gamma))$. The main theorem of renormalization theory can be now formulated on the level of graphs:
\be\label{Hopf:graph}
(S\circ Z)_\Gamma=\sum_{P\in \mathrm{Part}(V(\Gamma))}T_{\Gamma_P}\oc \bigotimes_{I\in P} Z_{\Gamma_I} 
\ee
where $Z_{\Gamma_I}:\mathcal{F}_{\mathrm{loc}}^{\otimes V(\Gamma_I)}\to (\mathcal{F}^{\otimes V(\Gamma_I)})_{\mathrm{loc}}$ and $\Gamma_P$ is the graph with vertex set $V(\Gamma_P)=V(\Gamma)$, with all lines connecting different index sets of the partition $P$, and $\Gamma_I$ is the graph with vertex set $V(\Gamma_I)=I$ and all lines of $\Gamma$ which connect two vertices in $I$.   To find the Hopf algebra structure underlying \eqref{Hopf:graph}, we have to go to the concrete renormalizable theory, where $\Rcal$ is mapped to $\tilde\Rcal$. 
More details can be found in \cite{Pin00b,BrouderFrabettiMenous2009,GBL00}. In non-renormalizable
theories, one has to use a generalization of Hopf algebras, which uses completed tensor products.

\section*{Conclusions and Outlook}
Causal perturbation theory is known to provide rigorous results on structural properties of renormalized 
perturbative quantum field theory in a transparent and elegant way. However, for models containing massless fields, the central 
solution%
\footnote{For a purely massive model the infrared behaviour is harmless and, hence,
one may choose $w_\gamma=\frac{x^\gamma}{\gamma!}$  in the $W$-projection
\eqref{W-projection}. This is the central solution of Epstein and Glaser \cite{EG73}, which maintains 
several symmetries (in particular Lorentz covariance), and is explicitly computable.
For the distribution splitting method \eqref{splitting0} in Minkowski space, it can easily be computed by a 
dispersion integral in momentum space \cite{EG73,Sch89}.}
 does not exist and a generally applicable method for explicit calculations is missing so far.

In this paper we develop such a method, by using dimensional regularisation in position space, proposed
by Bollini and Giambiagi \cite{BG96} some years ago.  More precisely, the regularization parameter is the 
index $\tfrac{d}2-1$ of the Bessel function
appearing in the Feynman propagator ($d$ denotes the spacetime dimension). Since, in the limit $\zeta\to 0$
(which removes the regularization) of the regularized time-ordered product $\TT^\zeta_n$,
there appears not only the overall divergence (localized on the thin diagonal $\Delta_n$),
but also subdivergences localized on partial diagonals, our method needs a 
position space version of Zimmermann's
Forest Formula, which adds suitable local counter terms in correct succession, such that the limit $\zeta\to 0$
exists. We prove such a formula (``Epstein-Glaser Forest Formula'', Thm.~\ref{forest}). It
is based on families of subsets of the set of vertices and not, as in Zimmermann's formula, on families of subgraphs.

Generally, Epstein-Glaser renormalization is non-unique. However, our regularized time-ordered products 
$\TT^\zeta_n$ are unique and, using the minimal subtraction prescription for the limit $\zeta\to 0$,
we get a unique result for the renormalized time-ordered products.

A main reason for the usefulness of conventional dimensional regularization is that the 
regularized time-ordered products are gauge invariant (in particular this holds true for the term
$\sim\zeta^0$ which is the minimal subtracted time-ordered product). To obtain
gauge invariance of our $\TT^\zeta_n$,  a crucial necessary condition is
that for all kinds of fields
the regularized Feynman propagator $\Delta_F^\zeta(x)$ is, for $x\not= 0$, a solution
of the pertinent free field equation. But, as we see from lemma \ref{mod-KG}, in case of a 
real scalar field, this holds only if we deform the Klein-Gordon operator into $d_\zeta=d-2\zeta$
dimensions. Hence, it seems that a $\zeta$-dependent deformation of the free Lagrangian is needed.
In \cite{FRb} gauge theories are incorporated into the Epstein-Glaser framework with the use of the so called Batalin-Vilkovisky formalism. This allows to keep track of gauge symmetry also in the regularized theory by means of the regularized quantum master equation (QME). We hope to apply these ideas also in the case of dimensional regularization.

The combinatorial structure we found can be described in Hopf algebraic terms.  The structure is similar to the structure found in the approach by Connes and Kreimer. There are, however, also differences. In particular, it turned out to be appropriate to distinguish carefully between the tensor product appearing in the decomposition of disconnected graphs and the composition of linear maps arising from finite renormalizations. In the one dimensional case these two products coincide, but in order to exploit  this fact one had to choose a basis and work with matrix elements.

\begin{appendix}
\section{Regularization in the Epstein-Glaser framework}\label{app:regularization}
The Epstein-Glaser method does neither involve any regularization nor divergent counter terms. Nevertheless, one may introduce a regularization and determine the necessary counter terms. This was already discussed in the original paper of Epstein and Glaser where Pauli-Villars regularization was used. With the help of the concept of the St\"uckelberg-Petermann renormalization group, this fact was formulated in \cite{BDF09} in the way that given a solution $S$ of the EG-axioms and a smooth approximation of the Feynman propagator $\Delta_F^\Lambda\to\Delta_F$ such  that the formal S-matrices $S_\Lambda$ can be directly defined, then there exists a sequence of renormalization group elements $\Zcal_\Lambda$ such that $S_\Lambda\circ \Zcal_\Lambda\to S$. The proof proceeds in the same way as the proof of the main theorem of renormalization \cite{DF04} and was not included in \cite{BDF09}. We therefore present it here in a slightly stronger form.
\begin{thm}
 Let $\Scal$ be a solution of the EG-axioms and let the Feynman propagator be approximated by a sequence of symmetric distributions 
$\Delta_F^\Lambda$ which converges in the H\"ormander topology with a scaling degree bounded by that of the Feynman propagator. Let 
$\Scal_\Lambda$ be a formal S-matrix associated to $\Delta_F^\Lambda$. Then there exists a sequence $\Zcal_\Lambda\in\mathscr{R}$ such that 
\[\Scal_\Lambda\circ \Zcal_\Lambda\to \Scal \ .\]  
\end{thm}
\begin{proof}
It is convenient to expand the formal S-matrix as a sum over graphs as explained in Section \ref{EG}.
We are going to show that
for each graph $\Gamma$ there exists a sequence of linear maps $Z_{\Gamma,\Lambda}:\mathcal{F}_{\mathrm{loc}}^{\otimes V(\Gamma)}\to (\mathcal{F}^{\otimes V(\Gamma)})_{\mathrm{loc}}$ such that
\[{T}_{\Gamma}=\lim_\Lambda \sum_{P\in\mathrm{Part}(V(\Gamma))}T_{\Gamma_P,\Lambda}\circ\bigotimes_{I\in P}Z_{\Gamma_I,\Lambda}\]
Here $\Gamma_P$ is the graph with vertex set $V(\Gamma_P)=V(\Gamma)$, with all lines connecting different index sets of the partition $P$, and $\Gamma_I$ is the graph with vertex set $V(\Gamma_I)=I$ and all lines of $\Gamma$ which connect two vertices in $I$.  
$Z_{\Gamma,\Lambda}$ is recursively defined by $Z_{\Gamma,\Lambda}=\mathrm{id}$ for graphs with one vertex (and no lines, since only graphs without tadpoles are admitted), $Z_{\Gamma,\Lambda}=0$ for EG-reducible graphs and by
\[Z_{\Gamma,\Lambda}=\langle t_{\Gamma,\Lambda},(\mathrm{id}-W_{\Gamma})\delta_{\Gamma}\rangle\]
for EG-irreducible graphs. Here $t_\Gamma$ contains already the contributions from subgraphs, i.e.
\[\langle t_{\Gamma,\Lambda},\delta_\Gamma\rangle=\sum_{|P|>1}T_{\Gamma_P,\Lambda}\circ\bigotimes_{I\in P}Z_{\Gamma_I,\Lambda}\ .\]
Due to the fact that $W_\Gamma$ coincides with the identity on elements of $\mathcal{D}(\Delta_\Gamma,Y_\Gamma)$ which vanish at the thin diagonal, $Z_{\Gamma,\Lambda}$
satisfies the locality condition 
\[Z_{\Lambda,\Gamma}((F+G)^{\otimes |V(\Gamma)|})=Z_{\Lambda,\Gamma}(F^{\otimes |V(\Gamma)|})+Z_{\Lambda,\Gamma}(G^{\otimes |V(\Gamma)|}) \]
for local functionals $F,G$ with disjoint support. 
$\Zcal_{\Lambda}$ is then defined by
\[\Zcal_{\Lambda}(F)=\sum_\Gamma \frac{1}{\mathrm{Sym}(\Gamma)}m_{|V(\Gamma)|}\circ Z_{\Gamma,\Lambda}(F^{\otimes |V(\Gamma)|})\]
where the sum extends over all graphs without tadpoles and with vertex sets $V(\Gamma)=\{1,\dots,n\}$ for some $n\in\mathbb{N}$.
$\Zcal_\Lambda$ then satisfies the locality condition and is thus an element of $\mathscr{R}$.

 \end{proof} 
\section{Minimal subtraction for the distribution-splitting method}\label{app:EG-splitting}
We assume that the reader is familiar with the distribution-splitting method
(see \cite{EG73,Sch89}). We recall: 
\begin{thm}\label{thm:splitting}
A distribution $d\in \Dcal^\prime(\RR^l)\ ,\,\,\,l:=d(n-1)$, with causal support,
$\supp\, d\subset (\bar V_+)^{\times (n-1)}\,\cup\,(\bar V_-)^{\times (n-1)}$,
has a unique solution $\bar{a}\in\Dcal_\lambda'(\RR^l)$, $\lambda:=\sd(d)-l$
of the splitting problem \eqref{splitting0}, that is the pointwise product
\be\label{splitting1}
 \bar a(x):=\Theta(v\cdot x)\,d(x)
\ee
exists in $\Dcal_\lambda'(\RR^l)$. Here, $\Theta$ has to be understood as the weak limit
$\Theta:=\lim_{\epsilon\downarrow 0}\,\chi_\epsilon$, where $(\chi_\epsilon)_{\epsilon>0}$ is
a family of smooth approximations of the Heaviside function with
$\supp\,\chi'\subset[0,\epsilon]$, and 
\be 
v\cdot x:=\sum_{j=1}^{n-1}v_j\cdot x_j\ , \,\quad v_j\in V_+\,\quad\forall j\ .
\ee
\end{thm}
Due to the causal support of $d$, the definition \eqref{splitting1} of
$\bar a$ is independent of the choice of $v_1,...,v_{n-1}\in V_+$.

With this Theorem, a solution $a\equiv a_W\in\Dcal'(\RR^l)$ of the splitting problem \eqref{splitting0}
can be obtained by means of a $W$-projection \eqref{W-projection}:
\be\label{splitting2}
\langle a_W,f\rangle :=\langle \bar a,Wf\rangle\ ,\quad
\quad \forall f\in \Dcal(\RR^l)\ .
\ee

Given a splitting solution $a\in\Dcal'(\RR^l)$, a solution 
$\dot t\in\Dcal'(\RR^l)$ of the corresponding\footnote{``Corresponding'' means that comparing the field 
expansions \eqref{causWick} of $\Dcal^\bzeta_n$ and $\TT^\bzeta_n\vert_{\text{outside} \Delta_n}$,
the numerical distributions $d$ and $t$ are the coefficients of the same field combination
$f^{\alpha_1}_{\beta_1}(\ph)(x_1)...f^{\alpha_n}_{\beta_n}(\ph)(x_n)$.} extension problem
$\Dcal^\prime(\RR^l\setminus\{0\})\ni t\rightarrow \dot t\in \Dcal^\prime(\RR^l)$
is obtained by $\dot t:=a-a'$. Note that $\mathrm{sd}(t)=\mathrm{sd}(d)$.
The distribution $a'\in\Dcal'(\RR^l)$ is inductively given by the time-ordered products
of lower orders, see \cite{EG73,Sch89}. Restricting to $\Dcal_\lambda(\RR^l)$, 
we conclude that the unique
splitting solution $\bar a$ (Theorem~\ref{thm:splitting}) and the unique
extension solution $\bar t$ (Theorem~\ref{thm:Extension-sd}) are related by
$\bar t:=\bar a-a'$.

Given functions $(w_b)_{|b|\leq\lambda}$ (where still
$\lambda:=\sd(d)-l\,$) determing a projection $W$ \eqref{W-projection}
and the pertinent splitting- and extension-solution
$a_W$ \eqref{splitting2} and $\dot t_W:=\bar t\circ W$, respectively, we define
a map $F$ from the set $S$ of solutions of the splitting problem,
\be\label{solutions-splitting}
S=\{a=a_W+\sum_{|b|\leq\lambda}C_b\,\partial^b\delta\,|\,C_b\in\CC\}\ ,
\ee
to the set $E$ of solutions of the extension problem,
\be\label{solutions-extension}
E=\{\dot t=\dot t_W+\sum_{|b|\leq\lambda}C_b\,\partial^b\delta\,|\,C_b\in\CC\}\ ,
\ee
by
\be\label{bijection}
F(a):=a-a'+\sum_{|b|\leq\lambda}\langle a',w_b\rangle\,(-1)^{|b|}\,\partial^b\delta\ .
\ee
Since 
\be\label{S->E}
F(a_W+\sum_{b}C_b\,\partial^b\delta)=\dot t_W+\sum_{b}C_b\,\partial^b\delta\ ,
\ee
the map $F$ is a {\it bijection}. To verify the latter equation, we use that 
$Ww_b=0\quad\forall |b|\leq\lambda$, hence $\langle a_W,w_b\rangle =0$ and
$\langle \dot t_W,w_b\rangle =0$. Since $F(a_w)$ \eqref{bijection} is an extension solution, it
can be written as $F(a_w)=\dot t_W+\sum_{|b|\leq\lambda}K_b\,\partial^b\delta$ and  
with that we obtain
\be
(-1)^{|c|}K_c=\langle F(a_W),w_c\rangle=-\langle a',w_c\rangle
+\sum_{|b|\leq\lambda}\langle a',w_b\rangle\,(-1)^{|b|}\,\langle\partial^b\delta,w_c\rangle =0\ ,
\ee
hence $F(a_W)=\dot t_W$, and this implies \eqref{S->E}.

Since for any $\dot t\in E$ there is a projection $W$ \eqref{W-projection}
with $\dot t=\dot t_W$, we conclude that any $a\in S$ is of the form 
$a=a_W$ for some projection $W$ \eqref{W-projection}.

Using these facts, our results about minimal subtraction (Sect.~\ref{MSnumerical})
and the differential renormalization formula \eqref{diffrenhom1} can be transformed to the splitting
problem as follows:
\begin{df}[Regularization] With the above notations, a family $\{a^\zeta\}_{\zeta
\in\Omega\setminus\{0\}}$, $a^\zeta\in\Dcal'(\RR^l)$, is a(n) (analytic / finite)
regularization of $d$, if $(a^\zeta-a')_{\zeta\in\Omega\setminus\{0\}}$ is
a(n) (analytic / finite) regularization of the corresponding $t\in\Dcal'(\RR^l\setminus\{0\})$.
\end{df} 
More explicitly, in the definition \ref{df:regularisation} the condition
\eqref{eq:regularization} is replaced by
\be
\lim_{\zeta\rightarrow0}\langle a^\zeta,g\rangle=\langle \bar{a},g\rangle\ ,\quad
\quad\forall g\in\Dcal_\lambda(\RR^l)\ .
\ee
Analogously to \eqref{regW-2}, we can write $\langle a_W,f\rangle$ as
\be
\langle a_W,f\rangle=\langle \bar{a},Wf\rangle=\lim_{\zeta\rightarrow0}
\left[\langle a^\zeta,f\rangle - \sum_{|b|\leq\lambda}
\langle a^\zeta,w_b\rangle\;\partial^b f(0)\right]\ ,\quad f\in\Dcal(\RR^l)\ .
\ee
Again, for an analytic regularization the principal parts of the two terms on the r.h.s.~must
cancel. Therefore, $\pp(a^\zeta)$ is a local distribution, 
\be
\pp(a^\zeta) = \sum_{|b|\leq\lambda}C_b(\zeta)\;\partial^b\delta\,,\quad\text{with}
\quad C_b(\zeta)=(-1)^{|b|}\pp(\langle a^\zeta,w_b\rangle)\ ,
\ee
and
\be
\langle a^\MS,f\rangle :=\lim_{\zeta\rightarrow0} \rp(\langle a^\zeta,f\rangle)
\ee
is a distinguished solution of the splitting problem \eqref{splitting0}
(the '$\MS$-solution').

{\it Differential renormalization} works also for the splitting problem \cite[Sect.~2.2]{Duet96}: let 
\be
d(x)=\partial_{r_1}...\partial_{r_{\omega+1}}d_1(x)\ ,\quad\quad\omega\in\NN_0\ ,
\ee
where $d_1$ has also causal support and $\mathrm{sd}(d_1)=\mathrm{sd}(d)-(\omega+1)<l$.
Then, $d_1$ can be splitted directly
(Theorem~\ref{thm:splitting}):
\be
a_1(x):=\Theta(v\cdot x)\,d_1(x)\in\Dcal'(\RR^l)\ .
\ee
With that a splitting solution $a$ of $d$ is obtained by
\be\label{splitting-diffren}
a(x)=\partial_{r_1}...\partial_{r_{\omega+1}}\bigl(\Theta(v\cdot x)\,d_1(x)\bigr)\ .
\ee

Assuming that $d^\bzeta$ scales homogeneously in $x$
with a non-integer degree $\kappa^\bzeta$, it follows that $d^\bzeta$
satisfies \eqref{diffrenhom} (with $P_0\equiv 0\equiv Q_0$). Hence, we can apply \eqref{splitting-diffren}
to split $d^\bzeta$:
\be
a^\bzeta(x)=\frac{1}{\prod_{k=0}^\omega (2{\bf l}\bzeta-k)}\,
\sum_{r_1...r_{\omega+1}}\partial_{r_1}...\partial_{r_{\omega+1}}
\Bigl(\Theta(v\cdot x)\,m_{x_{r_1}}...m_{x_{r_{\omega+1}}}\,d^\bzeta(x)\Bigr)\ .
\ee
Obviously, $a^\bzeta$ scales also homogeneously with degree $\kappa^\bzeta$; it
is the only splitting solution with this property. (The latter follows from the fact 
that $E$ \eqref{solutions-extension} contains precisely one homogeneous element and that
$S=E+a'$, taking into account that $a'$ is homogeneous with the same degree.) 
\\
\\
\end{appendix}

{\bf Acknowledgments.} We profitted a lot from stimulating discussions with Jos{\'e} M. Gracia-Bond{\'i}a. 
During working at this paper M.~D.~was mainly at the Max Planck Institute 
for Mathematics in the Sciences, Leipzig; he thanks Eberhard Zeidler for the invitations to 
Leipzig and for enlightening discussions.

{\small
\bibliographystyle{amsalpha}
\bibliography{Literatur}}
\end{document}